\documentclass{article}

\usepackage{microtype}
\usepackage{graphicx}
\usepackage{subfigure}
\usepackage{booktabs} 

\usepackage{hyperref}

\usepackage[accepted]{icml2021}

\usepackage{amsmath,amsthm,amsfonts,amssymb}
\usepackage{mathtools}
\usepackage{fancyhdr}
\usepackage{enumerate}
\usepackage{enumitem}
\usepackage{xfrac}
\usepackage{hyperref}
\usepackage[utf8]{inputenc}
\usepackage[english]{babel}
\usepackage{multicol}

\usepackage{soul}

\usepackage{color}

\icmltitlerunning{The Power of Randomization: Efficient and Effective Algorithms for Constrained Submodular Maximization}

\newtheorem{theorem}{Theorem}
\newtheorem{lemma}{Lemma}
\newtheorem{definition}{Definition}

\newcommand{\alglinelabel}{%
  \addtocounter{ALC@line}{-1}
  \refstepcounter{ALC@line}
  \label
}

\newcommand{\N}{\mathcal{N}}
\newcommand{\R}{\mathbb{R}}
\newcommand{\I}{\mathcal{I}}

\newcommand{\E}{\mathbb{E}}

\newcommand{\dom}{\mathrm{dom}}

\newcommand{\argmax}{\mathop{\arg\max}}
\newcommand{\avg}{\mathrm{avg}}
\newcommand{\opt}{\mathrm{opt}}
\newcommand{\abs}[1]{\lvert#1\rvert}
\newcommand{\A}{\mathcal{A}}
\newcommand{\RMG}{\textsc{RandomMultiGreedy}}
\newcommand{\ARG}{\textsc{AdaptRandomGreedy}}
\newcommand{\EV}{\mathcal{E}}

\newcommand{\BRG}{\textsc{BatchedRandomGreedy}}

\newcommand{\FSGS}{\textsc{FastSGS}}

\newcommand{\red}[1]{{\color{black}{#1}}}

\begin{document}

\onecolumn
    \icmltitle{The Power of Randomization: Efficient and Effective Algorithms\\ for Constrained Submodular Maximization}



    \icmlsetsymbol{equal}{*}

    \begin{icmlauthorlist}
        \icmlauthor{Kai Han}{ustc}
        \icmlauthor{Shuang Cui}{ustc}
        \icmlauthor{Tianshuai Zhu}{ustc}
        \icmlauthor{Jing Tang}{hkust}
        \icmlauthor{Benwei Wu}{ustc}
        \icmlauthor{He Huang}{suda}
    \end{icmlauthorlist}

    \icmlaffiliation{ustc}{School of Computer Science and Technology / SuZhou Research Institute, University of Science and Technology of China;}
    \icmlaffiliation{hkust}{Data Science and Analytics Thrust, The Hong Kong University of Science and Technology;}
    \icmlaffiliation{suda}{School of Computer Science and Technology, Soochow University}

    \icmlcorrespondingauthor{Kai Han}{hankai@ustc.edu.cn}

    \icmlkeywords{Machine Learning, ICML}

    \vskip 0.3in





\printAffiliationsAndNotice{}


\begin{abstract}
   Submodular optimization has numerous applications such as crowdsourcing and viral marketing. In this paper, we study the fundamental problem of non-negative submodular function maximization subject to a $k$-system constraint, which generalizes many other important constraints in submodular optimization such as cardinality constraint, matroid constraint, and $k$-extendible system constraint. The existing approaches for this problem achieve the best-known approximation ratio of $k+2\sqrt{k+2}+3$ (for a general submodular function) based on deterministic algorithmic frameworks. We propose several randomized algorithms that improve upon the state-of-the-art algorithms in terms of approximation ratio and time complexity, both under the non-adaptive setting and the adaptive setting. The empirical performance of our algorithms is extensively evaluated in several applications related to data mining and social computing, and the experimental results demonstrate the superiorities of our algorithms in terms of both utility and efficiency.
\end{abstract}

\section{Introduction}

Submodular optimization is an active research area in machine learning due to its wide applications such as crowdsourcing~\cite{singla2016noisy,HanTON2018}, clustering~\cite{gomes2010budgeted,HanPVLDB2019}, viral marketing~\cite{kempe2003maximizing,HanPVLDB2018}, and data summarization~\cite{badanidiyuru2014streaming,iyer2013submodular}. A lot of the existing studies in this area aim to maximize a submodular function subject to a specific constraint, and it is well known that these problems are generally NP-hard. Therefore, extensive approximation algorithms have been proposed, with the goal of achieving improved approximation ratios or lower time complexity.

Formally, given a ground set $\N$ with $|\N|=n$, a constrained submodular maximization problem can be written as:
\begin{eqnarray}
\max\{f(S): S\in \mathcal{I}\} \label{pb:basic}
\end{eqnarray}
where $f: 2^{\N}\mapsto \mathbb{R}_{\geq 0}$ is a submodular function satisfying $\forall X,Y\subseteq \N: f(X)+f(Y)\geq f(X\cup Y)+f(X\cap Y)$,
and $\mathcal{I}\subseteq 2^{\N}$ is the set of all feasible solutions. For example, if $\I=\{X: X\subseteq {\N}\wedge |X|\leq d\}$ for a given $d\in \mathbb{N}$, then $S\in \I$ represents a cardinality constraint. We also call $f(\cdot)$ ``monotone'' if it satisfies $\forall X\subseteq Y\subseteq \N: f(X)\leq f(Y)$, otherwise $f(\cdot)$ is called ``non-monotone''.

%

Although some application problems only have simple constraints like a cardinality constraint, many others have to be cast as submodular maximization problems with more complex ``independence system'' constraints such as matroid, $k$-matchoid, and $k$-system constraint. Among these constraints, the $k$-system constraint is the most general one, and a strict inclusion hierarchy of them is: cardinality~$\subset$~matroid~$\subset$~intersection of $k$ matroids~$\subset$~$k$-matchoid~$\subset$~$k$-extendible~$\subset$~$k$-system~\cite{Mestre2006}. 
Due to the generality of $k$-system constraint, it can be used to model a lot of constraints in various applications, such as graph matchings, spanning trees and scheduling~\cite{feldman2020simultaneous,mirzasoleiman2016fast}.

%
%
%

It is recognized that submodular maximization with a $k$-system constraint is one of the most fundamental problems in submodular optimization~\cite{calinescu2011maximizing,feldman2017greed,feldman2020simultaneous}, so a lot of efforts have been devoted to it since the 1970s, and the state-of-the-art approximation ratios are $k+1$~\cite{fisher1978analysis} and $k+2\sqrt{k+2}+3$~\cite{feldman2020simultaneous} for monotone $f(\cdot)$ and non-monotone $f(\cdot)$, respectively. \citet{feldman2020simultaneous} also showed that, by weakening their approximation ratio by a factor of $(1-2\epsilon)^{-2}$, their algorithm can be implemented under time complexity of $\mathcal{O}(\frac{kn}{\epsilon}\log(\frac{n}{\epsilon}))$. Surprisingly, all the existing algorithms with nice approximation ratios for this problem are intrinsically deterministic. Therefore, it is an interesting open problem whether the ``power of randomization'' can be leveraged to achieve better approximation ratios or better efficiency, as randomized algorithms are known to outperform the deterministic ones in many other problems.

It is noted that the utility function $f(\cdot)$ is assumed to be deterministic in Problem~\eqref{pb:basic}. However, in many applications such as viral marketing and sensor placement, the utility function could be stochastic and is only submodular in a probabilistic sense. To address these settings, \citet{golovin2011adaptive1} introduced the concept of adaptive submodular maximization, where each element $u\in \N$ is assumed to have a random state and the goal is to find an optimal \textit{adaptive policy} that can select a new element based on observing the realized states of already selected elements. Based on this concept, they also investigated the adaptive submodular maximization problem under a $k$-system constraint and provide an approximation ratio of $k+1$~\cite{golovin2011adaptive2}, but this ratio only holds when the utility function is \textit{adaptive monotone} (a property similar to the monotonicity property under the non-adaptive case). However, it still remains as an open problem whether provable approximation ratios can be achieved for this problem when the considered utility function is more general (i.e., not necessarily adaptive monotone).



In this paper, we provide confirmative answers to all the open problems mentioned above, by presenting novel randomized algorithms for the problem of (not necessarily monotone) submodular function maximization with a $k$-system constraint. Our algorithms advance the state-of-the-art under both the non-adaptive setting and the adaptive setting. More specifically, our contributions include:
\begin{itemize}
\item Under the non-adaptive setting, we present a randomized algorithm dubbed \textsc{RandomMultiGreedy} that achieves an approximation ratio of $(1+\sqrt{k})^2$ under time complexity of $\mathcal{O}(nr)$, where $r$ is the rank of the considered $k$-system. We also show that \textsc{RandomMultiGreedy} can be accelerated to achieve an approximation ratio of  $(1+\epsilon)(1+\sqrt{k})^2$ under nearly-linear time complexity of $\mathcal{O}(\frac{n}{\epsilon}\log \frac{r}{\epsilon})$. Therefore, our algorithm outperforms the state-of-the-art algorithm in~\cite{feldman2020simultaneous} in terms of both approximation ratio and time complexity. Furthermore, we show that \textsc{RandomMultiGreedy} can also be implemented as a deterministic algorithm with better performance bounds than the existing algorithms.

\item Under the non-adaptive setting, we also propose a randomized algorithm dubbed \BRG~which achieves a slightly worse approximation ratio of $(1+\epsilon)^2(1+\sqrt{k+1})^2$ but can be implemented in poly-logarithmic ``adaptive rounds''. This result greatly improves upon the best-known approximation ratio of $\frac{1+\epsilon}{(1-\epsilon)^2}(k+2\sqrt{2(k+1)}+5)$ achieved by the state-of-the-art algorithm with poly-logarithmic adaptivity proposed in~\cite{Francesco2021}.


\item Under the adaptive setting, we provide a randomized policy dubbed \textsc{AdaptRandomGreedy} that achieves an approximation ratio of $(1+\sqrt{k+1})^2$ when the utility function is not necessarily adaptive monotone. To the best of our knowledge, \textsc{AdaptRandomGreedy} is the first adaptive algorithm to achieve a provable performance ratio under this case. 

\item We test the empirical performance of the proposed algoirthms in several applications including movie recommendation, image summarization and social advertising with multiple products. The extensive experimental results demonstrate that, \textsc{RandomMultiGreedy} achieves approximately the same performance as the best existing algorithm in terms of utility, while its performance on efficiency is much better than that of the fastest known algorithm; besides, \ARG~can achieve better utility than the non-adaptive algorithms by leveraging adaptivity.
\end{itemize}

\vspace{-1.5ex}
For the fluency of description, the proofs of all our lemmas/theorems are deferred to the supplementary file.
\vspace{-1.5ex}

\section{Related Work} \label{sec:rw}


There are extensive studies on submodular maximization such as \cite{chekuri2019parallelizing, balkanski2019optimal,lee2010maximizing,Han2021,kuhnle2019interlaced}. For example, \citet{kuhnle2019interlaced} addressed a simple cardinality constraint using a nice ``interlaced greedy'' algorithm, where two candidate solutions are considered in a compulsory round-robin way; but it is unclear whether this algorithm can handle more complex constraints. In the sequel, we only review the studies most closely related to our work.



\textbf{Non-Adaptive Algorithms:} We first review the existing algorithms for non-monotone submodular maximization subject to a $k$-system constraint under the non-adaptive setting. The seminal work of~\cite{gupta2010constrained} proposed a \textsc{RepeatedGreedy} algorithm described as follows. At first, a series of candidate solutions $S_1,S_2,\cdots, S_{\ell}$ are sequentially found, where the elements of $S_j$ are greedily selected from $\N\setminus (\cup_{1\leq i\leq j-1}S_i)$ for all $j\in [\ell]$. After that, an Unconstrained Submodular Maximization (USM) algorithm (e.g., ~\cite{buchbinder2015tight}) is called to find $S_j'\subseteq S_j$ for all $j\in [\ell]$. Finally, the set in $\{S_j,S_j': j\in [\ell]\}$ with the maximum utility is returned. Note that the USM algorithm is only used as a ``black-box'' oracle and can be any deterministic/randomized algorithm, so this algorithmic framework is intrinsically deterministic. \citet{gupta2010constrained} showed that, by setting $\ell=k+1$, \textsc{RepeatedGreedy} can achieve an approximation ratio of $3k+6+{3}{k^{-1}}$ under $\mathcal{O}(nrk)$ time complexity. However, through a more careful analysis, \citet{mirzasoleiman2016fast} proved that \textsc{RepeatedGreedy} actually has an approximation ratio of $2k+3+k^{-1}$. Subsequently, \citet{feldman2017greed} further revealed that \textsc{RepeatedGreedy} can achieve an approximation ratio of $k+2\sqrt{k}+3+\frac{6}{\sqrt{k}}$ under $\mathcal{O}(nr\sqrt k)$ time complexity by setting $\ell=\lceil \sqrt{k}\rceil$.



\citet{han2020deterministic} proposed a different ``simultaneous greedy search'' framework, where two disjoint candidate solutions $S_1$ and $S_2$ are maintained simultaneously, and the algorithm always greedily selects a pair $(e,S_i)$ such that adding $e$ into $S_i$ brings the maximum marginal gain. By incorporating a ``thresholding'' method akin to that in~\cite{badanidiyuru2014fast}, \citet{han2020deterministic} proved that their algorithm achieves $(2k+2+\epsilon)$-approximation under $\mathcal{O}(\frac{n}{\epsilon}\log \frac{r}{\epsilon})$ time complexity. \citet{feldman2020simultaneous} also proposed an elegant algorithm where $\lfloor 2+\sqrt{k+2}\rfloor$ disjoint candidate solutions are maintained. By leveraging a thresholding method similar to~\cite{badanidiyuru2014fast,han2020deterministic}, \citet{feldman2020simultaneous} proved that their algorithm can achieve $(1-2\epsilon)^{-2}(k+2\sqrt{k+2}+3)$-approximation under $\mathcal{O}(\frac{kn}{\epsilon}\log \frac{n}{\epsilon})$ time complexity. On the hardness side, \citet{feldman2017greed} proved that no algorithm making polynomially many queries to the value and independence oracles can achieve an approximation better than $k+0.5-\epsilon$. For clarity, we list the performance bounds of the closely related algorithms mentioned above in Table~\ref{tab:overview}~\footnote{For simplicity, we only list the performance bounds of the accelerated versions of \citet{feldman2020simultaneous}'s algorithm and \textsc{RandomMultiGreedy} in Table~\ref{tab:overview}. Without acceleration, \citet{feldman2020simultaneous} can achieve an approximation ratio of $(1+\sqrt{k+2})^2$ under $\mathcal{O}({k}nr)$ time complexity, while \textsc{RandomMultiGreedy} achieves an approximation ratio of $(1+\sqrt{k})^2$ under $\mathcal{O}(nr)$ time complexity.}.

Recently, \citet{BalkanskiS18} proposed the concept of ``adaptivity'' for submodular optimization algorithms: an algorithm has $T$ adaptivity if it can be implemented in $T$ ``adaptive rounds'', where the algorithm is allowed to make polynomial number of independent queries to the function value of $f(\cdot)$ in each adaptive round. Based on this concept, a lot of studies (e.g.,  \cite{balkanski2019optimal,chekuri2019parallelizing,Matthew2019,Francesco2021}) have proposed algorithms with low-adaptivity under various constraints. Among these studies, \cite{Francesco2021} achieved the best performance bounds under a $k$-system constraint. More specifically, the \textsc{REP-SAMPLING} algorithm in \cite{Francesco2021} achieves an approximation ratio of $\frac{1+\epsilon}{(1-\epsilon)^2}(k+2\sqrt{2(k+1)}+5)$ with $\mathcal{O}(\frac{\sqrt{k}}{\epsilon^2}\log \frac{r}{k\epsilon}\log {n}\log r)$ adaptivity, while incurring at most $\mathcal{O}(\frac{\sqrt{k}n}{\epsilon^2}\log \frac{r}{k\epsilon}\log {n}\log r)$ queries to the function value of $f(\cdot)$ in expectation. Compared to \textsc{REP-SAMPLING}, our \BRG~algorithm achieves a much better approximation ratio of $(1+\epsilon)^2(1+\sqrt{k+1})^2$, under nearly the same adaptivity of $\mathcal{O}(\frac{\sqrt{k}}{\epsilon^2}\log \frac{r}{\epsilon}\log {n}\log r)$ and query complexity of $\mathcal{O}(\frac{\sqrt{k}n}{\epsilon^2}\log \frac{r}{\epsilon}\log {n}\log r)$.


\begin{table*}[t]
 \vspace{-2ex}
 \caption{Approximation for submodular function maximization with a $k$-system constraint}
  \vspace{-1ex}
 \label{tab:overview}
 \vskip 0.15in
 \begin{center}
  \begin{small}
   \begin{tabular}{lcccr}
    \toprule
    Algorithms & Source & Ratio & Time Complexity & Adaptive? \\
    \midrule
    \textsc{RepeatedGreedy}& \cite{gupta2010constrained}& $3k+6+{3}{k^{-1}}$& $\mathcal{O}(nrk)$& $\times$\\
    \textsc{RepeatedGreedy}& \cite{mirzasoleiman2016fast}& $2k+3+k^{-1}$& $\mathcal{O}(nrk)$& $\times$\\
    \textsc{TwinGreedyFast}& \cite{han2020deterministic}& $2k+2+\epsilon$& $\mathcal{O}(\frac{n}{\epsilon}\log(\frac{r}{\epsilon}))$& $\times$\\
    \textsc{RepeatedGreedy}& \cite{feldman2017greed}& $k+2\sqrt{k}+3+\frac{6}{\sqrt{k}}$& $\mathcal{O}(nr\sqrt k)$& $\times$\\
    \textsc{FastSGS}& \cite{feldman2020simultaneous} & $(1-2\epsilon)^{-2}(k+2\sqrt{k+2}+3)$ & $\mathcal{O}(\frac{kn}{\epsilon}\log(\frac{n}{\epsilon}))$ & $\times$ \\
    \textsc{RandomMultiGreedy}& this work & $(1+\epsilon)(k+2\sqrt k+1)$& $\mathcal{O}(\frac{n}{\epsilon}\log(\frac{r}{\epsilon}))$& $\times$\\
    \textsc{AdaptRandomGreedy}& this work & $k+2\sqrt{k+1} +2$ & $\mathcal{O}(nr)$ & $\surd$ \\
    \bottomrule
   \end{tabular}
  \end{small}
 \end{center}
 \vskip -0.1in
\end{table*}


\textbf{Adaptive Algorithms:} We then provide a brief review on the related studies on adaptive submodular maximization. \citet{golovin2011adaptive1} initiated the study on adaptive submodular maximization and also provided several algorithms under cardinality or knapsack constraints. They also studied the more general $k$-system constraint in~\cite{golovin2011adaptive2} and provided a $(k+1)$-approximation. Recently, \citet{Esfandiari2021adaptivity} proposed adaptive submodular maximization algorithms with fewer adaptive rounds of observation. There also exist many other studies on adaptive optimization under various settings/constraints, such as~\cite{cuong2016adaptive,mitrovic2019adaptive,parthasarathy2020adaptive,fujii2019beyond,badanidiyuru2016locally}. However, all these studies assumed that the target function is monotone or adaptive monotone. For non-monotone objective functions, \citet{amanatidis2020fast} and \citet{gotovos2015non} have proposed adaptive submodular maximization algorithms with provable performance ratios, but only under simple cardinality and knapsack constraints. 


\section{Preliminaries and Notations} \label{sec:pandn}


It is well known that all the structures including matroid, $k$-matchoid, $k$-extendible system and $k$-set system are set systems obeying the ``down-closed'' property captured by the concept of an \textit{independence system}: 

\begin{definition}[independence system]\label{definition:independence_system}
    Given a finite ground set $\N$ and a collection of sets $\I\subseteq 2^{\N}$, the pair $(\N,\I)$ is called an independence system if it satisfies: (1) $\emptyset\in\I$; (2) if $X\subseteq Y\subseteq\N$ and $Y\in\mathcal{I}$, then $X\in\mathcal{I}$.
\end{definition}

Given an independence system $(\N,\I)$ and any two sets $X\subseteq Y\subseteq \N$, $X$ is called a \textit{base} of $Y$ if $X\in \I$ and $X\cup \{u\}\notin \I$ for all $u\in Y\setminus X$. We also use $r$ to denote the \textit{rank} of $(\N,\I)$, i.e., $r=\max\{|X|: X\in \I\}$. A $k$-system is a special independence system defined as:
%

\begin{definition}[$k$-system]\label{definition:ksystem}
    An independence system $(\N,\I)$ is called a $k$-system ($k\geq1$) if $|X_1|\leq k|X_2|$ holds for any two bases $X_1$ and $X_2$ of any set $Y\subseteq \N$.
\end{definition}

\subsection{Non-adaptive Setting}
Under the non-adaptive setting, our problem is to identify an optimal solution $O$ to Problem~\eqref{pb:basic} given a $k$-system $(\N,\I)$ and a (not necessarily monotone) submodular function $f(\cdot)$. For convenience, we use $f(X\mid Y)$ as a shorthand for $f(X\cup Y)-f(Y)$ for all $X,Y\subseteq \N$. It is well known that any non-negative submodular function $f(\cdot)$ satisfies the ``diminishing returns'' property: $\forall X\subseteq Y\subseteq \N, x\in \N\setminus Y: f(x\mid Y)\leq f(x\mid X)$.
Following the existing studies, we assume that the values of $f(S)$ and $\mathbf{1}_{\I}(S)$ can be got by calling \textit{oracle queries}, and use the number of oracle queries to measure time complexity unless otherwise stated.

Following some related work such as \cite{balkanski2019optimal,chekuri2019parallelizing,Matthew2019,Francesco2021}, we define the ``adaptivity'' of an algorithm under the non-adaptive setting as the minimum number of ``adaptive rounds'' needed by the algorithm, such that the algorithm can make polynomially-many independent queries to the value of the objective function $f(\cdot)$ in each ``adaptive round''.

\subsection{Adaptive Setting}
Under the adaptive setting, each element $u\in \N$ is associated with an initially unknown state $\Phi(u)\in Z$, where $Z$ is the set of all possible states. A \textit{realization} is any function $\phi\colon \N\mapsto Z$ mapping every element $u\in \N$ to a state $z\in Z$. Therefore, $\Phi$ is the true realization and we follow \cite{golovin2011adaptive1} to assume that $\Pr [\Phi=\phi]$ is known for any possible realization $\phi$. In adaptive optimization problems, an \textit{adaptive policy} $\pi$ is allowed to sequentially select elements in $\N$, and the true state $\Phi(u)$ of any $u\in \N$ can only be observed after $u$ is selected. In such a case, the utility of $\pi$ depends on not only the selected elements but also their states, so we re-define the utility function as $f\colon 2^{\N}\times Z^{\N}\mapsto \R_{\geq0} $. Let $\N(\pi,\phi)$ denote the set of elements selected by $\pi$ under any realization $\phi$, the expected utility of policy $\pi$ is defined as
\begin{equation}\label{eq1}
    f_{\avg}(\pi) := \E [f(\N(\pi,\Phi),\Phi)], \nonumber
\end{equation}
where the expectation is taken over both the randomness of $\Phi$ and the internal randomness (if any) of $\pi$.

Given any $M\subseteq \N$, a mapping $\psi\colon M\mapsto Z$ is called a \textit{partial realization}, and $\dom(\psi)=M$ is called the \textit{domain} of $\psi$. Therefore, a partial realization is $\psi$ is also a realization when $\dom(\psi)=\N$. Intuitively, a partial realization can be used to record the already selected elements and the observed states of them during the execution of an adaptive policy. We also abuse the notations a little by regarding $\psi$ as the set $\{(u,\psi(u))\colon u\in \dom(\psi)\}$. Given two partial realizations $\psi$ and $\psi^\prime$, we say $\psi$ is a \textit{subrealization} of $\psi^\prime$ (denoted by $\psi'\sim \psi$) if $\psi\subseteq\psi^\prime$. With these definitions, we follow~\citet{golovin2011adaptive1} to define the concept of adaptive sumodularity:

\begin{definition}
    Given a partial realization $\psi$ and an element $u$, the expected marginal gain of $u$ conditioned on $\psi$ is defined as $\Delta(u\mid \psi)=\E[f(\dom(\psi)\cup \{u\},\Phi)-f(\dom(\psi),\Phi)\mid \Phi\sim\psi]$. A function $f\colon 2^{\N}\times Z^{\N}\mapsto \R_{\geq0} $ is called adaptive submodular if it satisfies $\forall \psi\subseteq \psi', u\in \N\setminus\dom(\psi^\prime)\colon \Delta(u\mid \psi)\geq \Delta(u\mid \psi^\prime)$.
\end{definition}


The utility function $f(\cdot)$ is also called \textit{adaptive monotone} if $\Delta(u\mid \psi)\geq 0$ for any $u\in \N$ and any partial realization $\psi$ satisfying $\Pr[\Phi\sim \psi]>0$. However, in this paper we consider the case that $f(\cdot)$ is not necessarily adaptive monotone. In such a case, all the the current studies (e.g., \cite{amanatidis2020fast,gotovos2015non}) assume that $f(\cdot)$ is also \textit{pointwise submodular}, whose definition is given below:

\begin{definition}
    A function $f\colon 2^{\N}\times Z^{\N}\mapsto \R_{\geq0} $ is pointwise submodular if $f(\cdot,\phi)$ is submodular for any realization $\phi$ satisfying $\Pr[\Phi=\phi]>0$.
\end{definition}

%

Given a $k$-system $(\N,\I)$ and an adaptive and pointwise submodular function $f\colon 2^{\N}\times Z^{\N}\mapsto \R_{\geq0}$, our problem is to identify an optimal policy $\pi_{\opt}$ to the following adaptive optimization problem:
\begin{equation}\label{eqn:adaptoptproblem}
    \max\{f_{\avg}(\pi)\colon \N(\pi,\phi)\in\I\text{~for~all~realization}~\phi\}.
\end{equation}


\section{A Randomized Algorithm under the Non-Adaptive Setting} \label{sec:nonadaptive}

In this section, we propose an algorithm dubbed \textsc{RandomMultiGreedy}, as shown by Algorithm~\ref{alg:muti-randomgreedy}. \textsc{RandomMultiGreedy} iterates for $T$ steps to construct $\ell$ candidate solutions $S_1, S_2, \cdots, S_{\ell}$. At each step $t$, it greedily finds a pair $(u_{t}, i_t)\in \N\times [\ell]$ such that $S_{i_t}\cup \{u_{t}\}\in \I$ and $f(u_{t}\mid S_{i_t})$ is maximized. If $f(u_{t}\mid S_{i_t})>0$, then Algorithm~\ref{alg:muti-randomgreedy} adds $u_t$ into $S_{i_t}$ with probability $p$ and discard $u_t$ with probability $1-p$. After that, $u_t$ is removed from $\N$. The iterations stop immediately when the pair $(u_{t}, i_t)$ cannot be found or $f(u_{t}\mid S_{i_t})\leq 0$.

For convenience, we introduce the following notations. Let $U=\{u_1,\cdots,u_T\}$ denote the set of all elements that have been considered to be added into $\cup_{i\in [\ell]}S_i$. For any $u\in \N$, let $S_i^<(u)$ denote the set of elements already in $S_i$ at the moment that $u$ is considered by the algorithm, and let $S_i^<(u)=S_i$ if $u$ is never considered by the algorithm. 

Although the design of \textsc{RandomMultiGreedy} is quite simple, its performance analysis is highly non-trivial due to the complex relationships between the elements in $S_1,\cdots, S_{\ell}$ and the randomness of the algorithm. To address these challenges, we first classify the elements in $O$ as follows:

\begin{algorithm}[t]
    \caption{\textsc{RandomMultiGreedy}$(\ell,p)$}
    \label{alg:muti-randomgreedy}
    \begin{algorithmic}
        \STATE {\bfseries Initialize:} $\forall i\in [\ell]:\ S_i\gets \emptyset;~~t\gets 1$
    \end{algorithmic}
    \begin{algorithmic}[1]
        \REPEAT
        \FOR{$i=1$ {\bfseries to} $\ell$} \alglinelabel{ln:greedystart}
        \STATE $A_i\gets \{u\in \mathcal{N}: S_i\cup \{u\}\in \mathcal{I} \}$
        \STATE $v_{i}\gets \arg\max_{u\in A_i}{f(u\mid S_i)}$ \alglinelabel{ln:greedyrule}
        \ENDFOR \alglinelabel{ln:greedyend}
        \IF{$\cup_{i\in [\ell]} A_i\neq \emptyset $}
         \STATE $i_t\gets \arg\max_{i\in [\ell]: A_i\neq \emptyset}f(v_{i}\mid S_i)$;~$u_t\gets v_{i_t}$
        \IF{$f(u_{t}\mid S_{i_t})> 0$}
        \STATE $\mathbf{with}$ probability $p$ $\mathbf{do}$ $S_{i_t}\gets S_{i_t}\cup \{u_{t}\}$ \alglinelabel{ln:addwithp} 
        \STATE $\mathcal{N}\leftarrow \mathcal{N}\setminus \{u_{t}\}; t\gets t+1$
        \ELSE
        \STATE \textbf{break};
        \ENDIF
                \ENDIF

        \UNTIL{$\bigcup_{i\in [\ell]}A_i= \emptyset\vee \N=\emptyset$}
        \STATE $S^*\gets \arg\max_{S\in \{S_1,S_2,\cdots,S_\ell \}} f(S); T\gets t-1$ \\
        \STATE {\bfseries Output:} $S^*, T$
    \end{algorithmic}
\end{algorithm}


\begin{definition} \label{def:keydef}
Let $D_j$ denote the set of elements in $U$ that have been considered to be added into $S_j$ but are discarded due to Line~\ref{ln:addwithp}. For any $i,j\in [\ell]$ satisfying $i\neq j$, we define:
    \begin{eqnarray*}
        &&O_j^{i+}=\left\{u\in O\cap S_j : S_i^<(u)\cup \{u\} \in \mathcal I \right\}; \\
        &&O_j^{i-}=\left\{u\in O\cap S_j : S_i^<(u)\cup \{u\} \notin \mathcal I \right\}; \\
        &&\widehat{O}_j^{i+}=\left\{u\in O\cap D_j : S_i^<(u)\cup \{u\} \in \mathcal I \right\}; \\
        &&\widehat{O}_j^{i-}=\left\{u\in O\cap D_j : S_i^<(u)\cup \{u\} \notin \mathcal I \right\};\\
        &&O_i^-=\left\{u\in O\setminus{U} : S_i\cup \{u\} \notin \mathcal I \wedge f(u\mid S_i)>0 \right\};
    \end{eqnarray*}

\end{definition}

Note that both $O_j^{i+}$ and $O_j^{i-}$ are disjoint subsets of $O\cap S_j$. Intuitively, each element $u\in O_j^{i+}$ (resp. $u\in O_j^{i-}$) can (resp. cannot) be added into $S_i$ without violating the feasibility of $\I$ at the moment that $u$ is added into $S_j$. The sets $\widehat{O}_j^{i+}$, $\widehat{O}_j^{i-}$ are also defined similarly for the elements in $O\cap D_j$. Based on Definition~\ref{def:keydef}, it can be seen that, when Algoirthm~\ref{alg:muti-randomgreedy} terminates, all the elements in $O_i^-$, $O_j^{i-}$ and $\widehat{O}_j^{i-}$ ($\forall j\neq i$) cannot be added into $S_i$ due to the violation of $\I$. Note that these elements together with the elements in $O\cap S_i$ all belong to $O$. So we can map them to the elements in $S_i$ using a method similar to that in~\cite{calinescu2011maximizing,han2020deterministic} based on the definition of $k$-system, as shown by Lemma~\ref{lma:mapping}:

\begin{lemma} \label{lma:mapping}
    For each $i\in [\ell]$, let $Q_i=\cup_{j\in [\ell]\setminus\{i\}}(O_j^{i-}\cup \widehat{O}_j^{i-})\cup (O\cap S_i)\cup O_i^-$. There exists a mapping $\sigma_i: Q_i \mapsto S_i$ satisfying: (1) The element $\sigma_i(u)$ can be added into $S_i^<({\sigma_i(u)})$ without violating the feasibility of $\I$ for all $u\in Q_i$; (2) The number of elements in $Q_i$ mapped to the same element in $S_i$ by $\sigma_i(\cdot)$ is no more than $k$; and (3) we have $\forall u\in O\cap S_i: \sigma_i(u)=u$.
%
%
\end{lemma}

The purpose for creating the mapping in Lemma~\ref{lma:mapping} is to bound the value of $f(u\mid S_i)$ for all $u\in Q_i$. For example, given any element $u\in O_j^{i-}$, $u$ can be mapped to an element $v=\sigma_i(u)$ satisfying $S_i^{<}(v)\cup \{u\}\in \I$, which implies that the value of $f(u\mid S_i)$ is no more than $f(v\mid S_i^<(v))$, because otherwise $u$ should have been added into $S_i$ instead of $v$ according to the greedy rule of Algorithm~\ref{alg:muti-randomgreedy}. Based on this intuition, a more careful analysis reveals that:




\begin{lemma} \label{lma:mgupperbound}
For any $u_t\in U$ where $t\in [T]$, we define $\delta(u_t)=\sum_{j=1}^\ell \mathbf{1}\{i_t=j\}\cdot f(u\mid S_{j}^<(u))$. Given any $i,j\in [\ell]$ satisfying $i\neq j$, we have
    \begin{eqnarray}
        &&\forall u\in O_j^{i+}\cup \widehat{O}_j^{i+}: f(u\mid S_i)\leq \delta(u);~~~\label{eqn:toprove1} \\
        &&\forall u\in Q_i: f(u\mid S_i)\leq \delta(\pi_i(u));~~~\label{eqn:toprove2}
    \end{eqnarray}
where $Q_i$ is defined in Lemma~\ref{lma:mapping}.
\end{lemma}

Using Lemma~\ref{lma:mgupperbound}, we can prove Lemma~\ref{lma:boundsumfosi}, which provides an upper bound of $\sum_{i\in[\ell]}f(O\mid S_i)$:

\begin{lemma} \label{lma:boundsumfosi}
For any $u\in \N$, define $X_u=1$ if $u\in (U\cap O)\backslash \cup_{i=1}^\ell S_i$, otherwise define $X_u=0$. Given any integer $\ell\geq 2$, we have
\begin{eqnarray}
\sum_{i\in[\ell]}f(O\mid S_i)\leq \ell(k + \ell -2) f(S^*)+\ell\sum_{u\in \N} X_u\cdot \delta(u) \label{eqn:keyeqn}
\end{eqnarray}
\end{lemma}

The proof idea of Lemma~\ref{lma:boundsumfosi} is roughly explained as follows. As the elements in $O\setminus S_i$ can be classified using the sets defined in Definition~\ref{def:keydef}, we can leverage Lemma~\ref{lma:mapping} and Lemma~\ref{lma:mgupperbound} to bound $f(u\mid S_i)$ for all $u\in O\setminus S_i$ using the marginal gains of the elements in $\cup_{i=1}^{\ell}S_i$. These marginal gains are further grouped in a subtle way such that their summation can be bounded by the RHS of Eqn.~\eqref{eqn:keyeqn}.


Note that Eqn.~\eqref{eqn:keyeqn} holds for \textit{every} random output of Algorithm~\ref{alg:muti-randomgreedy}. So the inequality still holds after taking expectation. Furthermore, we introduce Lemma~\ref{lma:boundexptation} to bound the expectations of the LHS and RHS of Eqn.~\eqref{eqn:keyeqn}. The proof of Lemma~\ref{lma:boundexptation} leverages the property that each element in $\N$ is only accepted with probability of at most $p$.

\begin{lemma} \label{lma:boundexptation}
For any $p\in (0,1]$, we have
    \begin{eqnarray}
        &&\mathbb{E} \bigg[\sum\nolimits_{i\in[\ell]}f(O\cup S_i)\bigg]\geq (\ell-p)f(O) \label{eqn:osiislarger}\\
        &&\mathbb{E} \bigg[\sum_{u\in\N}X_u\cdotp \delta(u)\bigg]\leq \frac{1-p}{p} \mathbb{E}\bigg[\sum_{i\in [\ell]}f(S_i)\bigg] \label{eqn:xudeltaissmall}
    \end{eqnarray}
\end{lemma}

By combining Lemma~\ref{lma:boundsumfosi}, Lemma~\ref{lma:boundexptation} and the fact that $\forall i\in [\ell]: f(S_i)\leq f(S^*)$, we can immediately get the approximation ratio of Algorithm~\ref{alg:muti-randomgreedy} as follows:


\begin{theorem} \label{thm:rmgapproxratio}
For any $\ell\geq 2$ and $p\in (0,1]$, the \textsc{RandomMultiGreedy}$(\ell,p)$ algorithm outputs a solution $S^*$ satisfying
\begin{eqnarray}
f(O) \leq \frac{\ell(k+\frac{\ell}{p}-1)}{\ell-p}\mathbb{E} [f(S^*)]
\end{eqnarray}
\end{theorem}

\textbf{Discussion of Theorem~\ref{thm:rmgapproxratio}:} From Theorem~\ref{thm:rmgapproxratio}, it can be seen that the approximation ratio of Algorithm~\ref{alg:muti-randomgreedy} can be optimized by choosing proper values of $\ell$ and $p$. Indeed, the ratio can be minimized to $(1+\sqrt{k})^2$ by setting $\ell=2,p=\frac{2}{1+\sqrt{k}}$. Besides, if we set $\ell=\lceil\sqrt{k}\rceil+1,p=1$, then the approximation ratio turns into $k+\sqrt{k}+\lceil \sqrt{k}\rceil+1$. Clearly, setting $\ell=2$ implies faster running time as only two candidate solutions are maintained, while setting $p=1$ implies a deterministic algorithm.


\subsection{Acceleration} \label{sec:accel}

It can be seen that \textsc{RandomMultiGreedy} has time complexity of $\mathcal{O}(\ell nr)$. This time complexity can be further reduced by implementing Lines~\ref{ln:greedystart}-\ref{ln:greedyend} using a ``lazy evaluation'' method inspired by~\cite{minoux1978accelerated,ene2019nearly}. More specifically, for each solution set $S_i$, we maintain an ordered list $A_i$ which is initialized to $\N$. Each element $u\in A_i$ has a weight $w_i(u)=f(u\mid S_i)$ and the elements in $A_i$ are always sorted according to the non-increasing order of their weights. When $S_i$ changes, we pop out the top element $u$ from $A_i$ and discard $u$ if $S_i\cup \{u\}\notin \I$. If $S_i\cup \{u\}\in \I$ and $f(u\mid S_i)$ has not been computed, then we update the weight of $u$ and set $v_i=u$ if the new weight of $u$ is at least $(1+\epsilon)^{-1}$ fraction of its old weight (otherwise $u$ is re-inserted into $A_i$ and we pop out the next element). During this process, any element in $A_i$ is removed from $A_i$ immediately when its weight has been updated for more than $ \mathcal{O}(\frac{1}{\epsilon}\log \frac{\ell r}{\epsilon})$ times. Using this method, we can guarantee that $f(v_i\mid S_i)$ is at least $\frac{1}{1+\epsilon}$ fraction of the marginal gain of the best element in $A_i$ that can be added into $S_i$, and the total number of  incurred value and independence oracles is no more than $\mathcal{O}(\frac{n}{\epsilon}\log\frac{\ell r}{\epsilon})$ for each $S_i: i\in [\ell]$. Combining these results with Theorem~\ref{thm:rmgapproxratio}, we can get:

\begin{theorem} \label{thm:approxratioacc}
For the problem of submodular maximization subject to a $k$-system constraint, there exist: (1) a randomized algorithm with an approximation ratio of $(1+\epsilon)(1+\sqrt{k})^2$ under $\mathcal{O}(\frac{n}{\epsilon}\log\frac{r}{\epsilon})$ time complexity, and (2) a deterministic algorithm with an approximation ratio of $(1+\epsilon)(k+\sqrt{k}+\lceil\sqrt{k}\rceil+1)$ under $\mathcal{O}(\frac{\sqrt{k}n}{\epsilon}\log \frac{\sqrt{k}r}{\epsilon})$ time complexity. 
\end{theorem}

\textbf{Remark:} From Theorem~\ref{thm:approxratioacc}, it can be seen that Algorithm~\ref{alg:muti-randomgreedy}~actually can be regarded as a ``universal algorithm'' that achieves the best-known performance bounds under different settings, as explained in the following. First, if $f(\cdot)$ is non-monotone, then Algorithm~\ref{alg:muti-randomgreedy} outperforms the state-of-the-art algorithm of \cite{feldman2020simultaneous} in terms of both approximation ratio and time efficiency, no matter Algorithm~\ref{alg:muti-randomgreedy} is implemented as a randomized algorithm or as a deterministic algorithm; moreover, when $(\N, \I)$ is a matroid (i.e., $k=1$), Algorithm~\ref{alg:muti-randomgreedy} achieves an approximation ratio of $4+\epsilon$ under $O(\frac{n}{\epsilon}\log\frac{r}{\epsilon})$ time complexity, matching the performance bounds of the fastest algorithm in~\cite{han2020deterministic} for a matroid constraint. Second, if the considered submodular function $f(\cdot)$ is monotone, it can be easily seen that the standard greedy algorithm proposed in \cite{fisher1978analysis} equals to $\textsc{RandomMultiGreedy}(1,1)$, so Algorithm~\ref{alg:muti-randomgreedy} also achieves the best-known approximation ratio of $k+1$ under this case.


\section{A Randomized Non-Adaptive Algorithm with Lower Adaptivity} \label{sec:brg}

It can be easily seen that the \textsc{RandomMultiGreedy} algorithm has $\mathcal{O}(r)$-adaptivity, as only one element is added into the candidate sets at each time. As $r$ can be in the order of $O(n)$, \textsc{RandomMultiGreedy}~has a large adaptivity. In this section, we propose a new algorithm dubbed \textsc{BatchedRandomGreedy} with lower adaptivity, as shown by Algorithm~\ref{alg:brg}.

In contrast to \RMG, the \BRG~algorithm maintains only one candidate solution $S$, and uses a threshold $\tau$ ranges from $\tau_{max}$ to $\tau_{min}$ to control the quality of the elements added into $S$. Given a threshold $\tau$, \RMG~first finds a set $C$ of elements which have not been considered yet, such that each element in $C$ can be added into $S$ without violating $\I$ to bring a marginal gain no less than $\tau$, and then runs in iterations and and tries to add a batch of elements into $S$ in each iteration, as shown by Lines~\ref{ln:whilestart}-\ref{ln:whileend}. More specifically, in each iteration, Algorithm~\ref{alg:brg} first selects a random sequence $\{a_1,\cdots, a_d\}$ by calling the \textsc{RanSEQ} procedure, and then uses binary search to find a sub-sequence $\{a_1,\cdots, a_j\}$ of $\{a_1,\cdots, a_d\}$ such that at least a $(\frac{1}{1+\epsilon})$-fraction of the elements in $C$ can be dropped if $\{a_1,\cdots, a_j\}$ is added into $S$; finally, it adds $\{a_1,\cdots, a_j\}$ into $S$ with probability of $p$ and updates $C$ accordingly. These iterations terminates when $C=\emptyset$, and then \BRG~reduces $\tau$ by a factor of $\frac{1}{1+\epsilon}$ and repeats the above process.

Note that the \textsc{RandSEQ} procedure is inspired by \citet{KarpUW88}, whose goal is to find a subset $A$ of $C$ with maximum cardinality such that $S\cup A\in \I$. It adopts a simple idea of repeatedly adding a random sequence of elements into $A$ until no more elements can be added without violating $\I$. \citet{KarpUW88} have proved that \textsc{RandSEQ} can terminate by repeating the while-loop in it for at most $\mathcal{O}(\sqrt{n})$ times.
\begin{algorithm} [!t]
    \caption{\textsc{RandSEQ}$(S,C)$}
    \label{alg:randseq}
    \begin{algorithmic}
        \STATE {\bfseries Initialize:} $A\gets\emptyset$
    \end{algorithmic}
    \begin{algorithmic}[1]
        \WHILE{$C\neq \emptyset$}
            \STATE Select a random permutation $z_1,\cdots,z_{|C|}$ of the elements in $C$
            \STATE $s\gets \max\{i\in [|C|]: S\cup A\cup \{z_1,\cdots,z_s\}\in \I\}$
            \STATE $A\gets A\cup \{z_1,\cdots,z_s\}$;\
            \STATE $C\gets \{u: u\in C\backslash A\wedge S\cup A\cup \{u\}\in \I\}$;\
        \ENDWHILE
        \STATE \textbf{return}\ $A$\;
    \end{algorithmic}
\end{algorithm}

\begin{algorithm} [!t]
    \caption{\textsc{BatchedRandomGreedy}$(p,\epsilon)$}
    \label{alg:brg}
    \begin{algorithmic}
        \STATE {\bfseries Initialize:} $\tau_{max}\gets \max\{f(u\mid \emptyset):u\in \N\};~\tau_{min}\gets \epsilon\cdot\tau_{max}/r; S\gets\emptyset; U\gets \emptyset$
    \end{algorithmic}
    \begin{algorithmic}[1]

        \FOR{$(\tau\gets \tau_{max};~\tau\geq\tau_{min};~\tau\gets \frac{\tau}{1+\epsilon})$}

        \STATE{$C\gets \{u\in \N\backslash U: S\cup \{u\}\in \I\wedge f(u\mid S)\geq \tau\}$}
        \WHILE{$C\neq \emptyset$}
            \STATE $\{a_1,a_2,\cdots, a_d\}\gets \mathsf{RandSEQ}(S,C)$ \alglinelabel{ln:whilestart}
            \STATE{Define $C_i= \{u: u\in C\wedge S\cup \{a_1,\cdots,a_{i}, u\}\in \I\wedge f(u\mid S\cup \{a_1,a_2,\cdots, a_{i}\})\geq \tau\}$ for each $i\in \{0,\cdots,d\}$}
            \STATE{Using binary search to find $j= \min\left\{i\in [d]: |C_i|< \frac{|C|}{1+\epsilon}\right\}$} \alglinelabel{ln:definejmini}
             \STATE{$U\gets U\cup \{a_1,\cdots,a_j\}$}
             \STATE\textbf{with}\ probability $p$ \textbf{do}
                \STATE\quad $S\gets S\cup \{a_1,\cdots,a_j\};~C\gets C_{j}$;\
             \STATE\textbf{otherwise}
                \STATE\quad $C\gets C_0\backslash U$; \alglinelabel{ln:whileend}
        \ENDWHILE
        \ENDFOR
        \STATE \textbf{return}\ $S$\;

    \end{algorithmic}

\end{algorithm}

\subsection{The Approximation Ratio of \BRG}
In this section, we analyze the approximation ratio of the \BRG~algorithm. For clarity, we first introduce the following notations. When the algorithm finishes, let $\{e_1,e_2,\cdots, e_{r}\}$  denote the set of elements sequentially added into $S$, and let $S_i$ denote the set $\{e_1,\cdots,e_i\}$ for any $i\in [r]$. As it is possible that $|S|<r$, we define $e_t$ as a ``dummy element'' satisfying $f(e_t\mid S_{t-1})=0$ if $|S|<t\leq r$. Similar to that in Sec.~\ref{sec:nonadaptive}, for any $u\in U$, we use $S^<(u)$ to denote the set of elements already added into $S$ before considering $u$ (note that $U$ is defined in Algorithm~\ref{alg:brg}, denoting the set of all elements considered by the algorithm to be added into $S$). For any $\tau$ considered in the for-loop of Algorithm~\ref{alg:brg}, let $S^{\tau}\subseteq S$ denote the current set of elements in $S$ right after the while-loop using $\tau$ in Algorithm~\ref{alg:brg} is finished. Finally, we use $O^-$ to denote the set $\left\{u\in O\setminus{U} : S\cup \{u\} \notin \mathcal I \wedge f(u\mid S)>0 \right\}$.

With the above definitions, we first introduce the following lemma, which implies that \BRG~is intrinsically a greedy algorithm that iteratively selects an element with the (approximately) maximal marginal gain in expectation:
\begin{lemma}
We have
\begin{eqnarray}
\forall i\in [k]:~~(1+\epsilon)^2\E[f(e_i\mid S_{i-1})\mid S_{i-1}]\geq {\max\{f(u\mid S_{i-1})\mid u\in \N\backslash U\wedge S_{i-1}\cup \{u\}\in \I\}} 
\end{eqnarray}
\label{lma:boundingmarginalgain}
\end{lemma}
With Lemma~\ref{lma:boundingmarginalgain}, we can prove the approximation ratio of \BRG~using the ideas described as follows. When the \BRG~algorithm terminates, any element $u\in O\backslash S$ must belong to one of the following three categories: (1) $u\in O^-$; (2) $u\in (U\cap O)\backslash S$; (3) $f(u\mid S)\leq \tau_{min}$. So we prove the approximation ratio by bounding the total marginal gains of the elements in each category with respect to $S$. More specifically, for the elements in the first category, this bounding can be done by creating a mapping similar to that in Lemma~\ref{lma:mapping}; for the elements in the second category, we can bound their marginal gains by leveraging the property that these elements are discarded with probability of $1-p$; the marginal gains of the elements in the third categories can be bounded by using the definition of $\tau_{min}$. Based on these ideas, we prove the following theorem:

\begin{theorem}
The \BRG~algorithm returns a solution $S$ satisfying
\begin{eqnarray}
f(O)\leq \frac{(1+\epsilon)^2 k+\frac{1}{p}+\epsilon}{1-p}\E[f(S)],
\end{eqnarray}
which implies that it achieves an approximation ratio of $(1+\epsilon)^2(1+\sqrt{k+1})^2$ if $p$ is set to $(1+\sqrt{k+1})^{-1}$.
\label{thm:ratioofbrg}
\end{theorem}

\subsection{Complexity Analysis on \BRG}

Roughly speaking, \BRG~has a low adaptivity due to the following reasons. First, only independent value oracle queries to $f(\cdot)$ are needed to find each batch of elements that \BRG~tries to add into $S$. Second, the number of batches of elements considered by \BRG~is logarithmic with respect to $n$, as \BRG~repeatedly drops a $(\frac{1}{1+\epsilon})$-fraction of the candidate elements. So we get the following theorem:




%


\begin{lemma}
The \BRG~algorithm can be implemented in $\mathcal{O}(\frac{\sqrt{k}}{\epsilon^2}\log \frac{r}{\epsilon}\log {n}\log r)$ adaptive rounds in expectation, and incurs an expected number of $\mathcal{O}(\frac{\sqrt{k}n}{\epsilon^2}\log \frac{r}{\epsilon}\log {n}\log r)$ oracle queries to the function value of $f(\cdot)$.
\label{lma:complexityofadaptivity}
\end{lemma}
\begin{proof}
Note that there are at most $\mathcal{O}(\frac{1}{\epsilon}\log \frac{r}{\epsilon})$ different values of $\tau$ considered in the for-loop of \BRG. Besides, for each tested threshold $\tau$, the size of the set $C$ of candidate elements gets smaller by at least a factor of $\frac{1}{1+\epsilon}$ in at most $\sum_{i=1}^\infty i\cdot (1-p)^{i-1}\cdot p=p^{-1}$ adaptive rounds in expectation (due to the reason that the considered sequence can be dropped with a probability of $p$), which implies that there are at most $\mathcal{O}(\frac{1}{p\epsilon}\log n)$ iterations in the while-loop for threshold $\tau$. Note that each iteration of the while-loop of \BRG~can be implemented in $\mathcal{O}(\log r)$ adaptive rounds due to the binary search process in Line~\ref{ln:definejmini}. So the claimed complexity on adaptivity follows by combining all the above results and by setting $p=(1+\sqrt{k+1})^{-1}$ (the same value as that in Theorem~\ref{thm:ratioofbrg}). Finally, the claimed query complexity holds due to the reason that at most $\mathcal{O}(n)$ value oracle queries are incurred in each adaptive round of \BRG.
\end{proof}

Note that \BRG~also incurs independence oracle queries (i.e., query the value of $\mathbf{1}_{\I}(X)$ for any given set $X$). These independence oracle queries can also be implemented in parallel if they are mutually independent. So we can investigate the number of ``adaptive rounds on independence oracle queries'' of \BRG~in a way similar to that in Lemma~\ref{lma:complexityofadaptivity}, as shown by the following lemma:

\begin{lemma}
In expectation, the \BRG~algorithm can be implemented in $\mathcal{O}(\frac{\sqrt{kn}}{\epsilon^2}\log \frac{r}{\epsilon}\log {n}\log r)$ adaptive rounds on independence oracle queries, and incurs no more than $\mathcal{O}(\frac{n^{3/2}\sqrt{k}}{\epsilon^2}\log \frac{r}{\epsilon}\log {n}\log r)$ independence oracle queries.
\end{lemma}
\begin{proof}
In Lemma~\ref{lma:complexityofadaptivity}, we have proved that the while-loop of \BRG~has at most $\mathcal{O}(\frac{\sqrt{k}}{\epsilon^2}\log \frac{r}{\epsilon}\log {n})$ iterations in expectation. In each iteration of the while-loop, the \textsc{RandSEQ} procedure is called once, which can be implemented in at most $\mathcal{O}(\sqrt{n})$ adaptive rounds for independence oracle queries according to \citet{KarpUW88}; after that, the binary search process in Line~\ref{ln:definejmini} can be implemented in at most $\mathcal{O}(\log r)$ adaptive rounds. Finally, it is noted that at most $\mathcal{O}(n)$ independence oracle queries can be incurred in each adaptive round. So the lemma follows by combining all the above results.
\end{proof}

\section{Adaptive Optimization} \label{sec:adaptopt}


The framework of Algorithm~\ref{alg:muti-randomgreedy} can be naturally extended to address the adaptive case (i.e., Problem~\eqref{eqn:adaptoptproblem}), as shown by Algorithm~\ref{alg:arg}. For convenience, we use $\pi_{\A}$ to denote the adaptive policy adopted by Algorithm~\ref{alg:arg}. Algorithm~\ref{alg:arg} runs in iterations and identifies an element $u^*$ in each iteration which maximizes the expected marginal gain $\Delta(u^*\mid \psi)$ without violating the feasibility $\I$, where $\psi$ is the partial realization observed by $\pi_{\A}$ at the moment that $u^*$ is identified. After that, $\pi_{\A}$ observes the state of $u^*$ and adds $u^*$ into the solution set $S$ with probability $p$, and discard $u^*$ with probability $1-p$. The algorithm stops when no more elements can be added into $S$ without violating the feasibility of $\I$ or when $\Delta(u^*\mid \psi)$ is non-positive.

Although the framework of Algorithm~\ref{alg:arg} looks similar to \textsc{RandomMultiGreedy}, its performance analysis is very different, as there does not exist a fixed optimal solution set under the adaptive setting, and we have to compare the average performance of $\pi_{\A}$ with that of an optimal policy $\pi_{\opt}$. To address this problem, we first build a relationship between $\pi_{\A}$ and $\pi_{\opt}$ as follows:

\begin{algorithm} [!tbp]
    \caption{\textsc{AdaptRandomGreedy}$(p)$}
    \label{alg:arg}
    \begin{algorithmic}
        \STATE {\bfseries Initialize:} $S\gets \emptyset$ and $\psi\gets \emptyset$
    \end{algorithmic}
    \begin{algorithmic}[1]

        \WHILE{$\N\neq \emptyset$}
        \STATE $A\gets \{u\in \N: S\cup \{u\}\in \I\}$\;
        \STATE $u^\ast\gets \argmax_{u\in A} \Delta(u\mid \psi)$\;\alglinelabel{algline:greedy}
        \IF{$A=\emptyset\vee \Delta(u^\ast\mid \psi)\leq 0$}
        \STATE\textbf{break}\alglinelabel{algline:negative}
        \ENDIF
        \STATE\textbf{with}\ probability $p$ \textbf{do}\alglinelabel{algline:prob}
        \STATE\quad observe $z=\Phi(u^\ast)$;\
        \STATE\quad $S\gets S\cup\{u^\ast\}$;\alglinelabel{algline:select}
        \STATE\quad $\psi \gets \psi\cup \{\big(u^\ast,z\big)\}$\;\alglinelabel{algline:observe}
        \STATE $\N\gets \N\setminus\{u^*\}$\;
        \ENDWHILE
        \STATE \textbf{return}\ $S$\;

    \end{algorithmic}

\end{algorithm}

\begin{lemma}\label{lma:concat}
Given any two adaptive policy $\pi_1$ and $\pi_2$, let $\pi_1@\pi_2$ denote a new policy that first execute $\pi_1$ and then execute $\pi_2$ without any knowledge about $\pi_1$. So we have 
    \begin{eqnarray}
     f_{\avg}(\pi_{\A}@\pi_{\opt})=f_{\avg}(\pi_{\opt}@\pi_{\A})\geq (1-p)\cdot f_{\avg}(\pi_{\opt}) \nonumber
    \end{eqnarray}
\end{lemma}

Lemma~\ref{lma:concat} implies that we may get an approximation ratio by further bounding $f_{\avg}(\pi_{\A}@\pi_{\opt})$ using $f_{\avg}(\pi_{\A})$. Given any $u\in \N$ and any realization $\phi$, let $\psi_u(\phi)$ denote the partial realization observed by $\pi_{\A}$ right before $u$ is considered by Lines~\ref{algline:prob}-\ref{algline:observe} of Algorithm~\ref{alg:arg}; if $u$ is never considered, then let $\psi_u(\phi)$ denote the observed partial realization at the end of $\pi_{\A}$. Based on this definition, we can get:

\begin{lemma} \label{lma:upperboundofmg}
The value of $f_{\avg}(\pi_{\A}@\pi_{\opt})-f_{\avg}(\pi_{\A})$ is no more than $\E_{\pi_{\A},\Phi}\Big[\sum_{u\in \N(\pi_{\opt},\Phi)\backslash \N(\pi_{\A},\Phi)}\Delta(u\mid \psi_{u}(\Phi))\Big]$, where the expectation is taken with respect to both the randomness of $\Phi$ and the randomness of $\pi_{\A}$.
\end{lemma}

Next, we try to establish some quantitative relationships between $f_{\avg}(\pi_{\A})$ and the upper bound found in Lemma~\ref{lma:upperboundofmg}. Given any realization $\phi$, Note that $\N(\pi_{\opt},\phi)\backslash \N(\pi_{\A},\phi)$ denotes the set of elements that are selected by $\pi_{\opt}$ but not $\pi_{\A}$ under the realization $\phi$. The elements in this set can be partitioned into three disjoint sets $O_1(\phi),O_2(\phi)$ and $O_3(\phi)$, where $O_2(\phi)$ denotes the set of elements that have been considered by $\pi_{\A}$ in Lines~\ref{algline:prob}-\ref{algline:observe} but discarded (due to the probability $p$); $O_3(\phi)$ denotes the set of elements satisfying $\Delta(u\mid \psi_{u}(\phi))\leq 0$ for all $u\in O_3(\phi)$; and the rest elements are all in $O_1(\phi)$. It can be seen that each element $u$ in $O_1(\phi)$ must satisfy $\dom(\psi_{u}(\phi))\cup \{u\}\notin \I$. Therefore, by using a similar method as that under the non-adaptive case, we can map the elements in $O_1(\phi)$ to the elements selected by $\pi_{\mathcal{A}}$ under realization $\phi$, and hence prove:


\begin{lemma}\label{lma:boundo1}
    We have $$\E_{\pi_{\A},\Phi} \Big[\sum\nolimits_{u\in O_1(\Phi)}\Delta(u\mid \psi_{u}(\Phi))\Big] \leq k\cdot f_{\avg}(\pi_{\A})$$

\end{lemma}

Now we try to bound the ``utility loss'' caused by $O_2(\phi)$. Note that although these elements are discarded (with probability $1-p$), they got a chance to be selected by $\pi_{\A}$ with probability $p$. So the ratio of the total expected (conditional) marginal gain of these elements to $f_{\avg}(\pi_{\A})$ should be no more than $(1-p)/p$, which is proved by the following lemma:

\begin{lemma}\label{lma:boundo2}
    We have
    \begin{equation}
        \E_{\pi_{\A},\Phi}\bigg[\sum_{u\in O_2(\Phi)} \Delta(u\mid \psi_{u}(\Phi))\bigg]\leq \frac{1-p}{p}\cdot f_{\avg}(\pi_{\A}) \nonumber
    \end{equation}

\end{lemma}

Combining all the above lemmas, we can get the approximation ratio of \textsc{AdaptRandomGreedy} as follows:

\begin{theorem} \label{thm:argapproxratio}
\textsc{AdaptRandomGreedy} achieves an approximation ratio of $\frac{pk+1}{p(1-p)}$ (i.e., $f_{\avg}(\pi_{\A})\geq \frac{p(1-p)}{pk+1}\cdot f_{\avg}(\pi_{\opt})$) under time complexity of $\mathcal{O}(nr)$. The ratio is minimized to $(1+\sqrt{k+1})^2$ when $p=(1+\sqrt{k+1})^{-1}$.
\end{theorem}

\textbf{Remark:} When the objective function $f(\cdot)$ is monotone, it can be easily seen that \textsc{AdaptRandomGreedy}$(1)$ can achieve an approximation ratio of $(k+1)$--the same ratio as that in~\cite{golovin2011adaptive2}. Therefore, \textsc{AdaptRandomGreedy}$(p)$ can also be considered as a ``universal algorithm'' for both non-monotone and monotone submodular maximization.

\section{Performance Evaluation} \label{sec:pe}

In this section, we compare our algorithms with the state-of-the-art algorithms for submodular maximization subject to a $k$-system constraint, using the metrics of both utility and the number of oracle queries to the objective function. We implemented five algorithms in the experiments: (1) the accelerated version of our \RMG~algorithm (as described in Sec.~\ref{sec:accel}), abbreviated as ``RAMG''; (2) the \textsc{RepeatedGreedy} algorithm presented in \cite{feldman2017greed}, abbreviated as ``REPG''; (3) the \textsc{TwinGreedyFast} algorithm proposed in \cite{han2020deterministic}, abbreviated as ``TGF''; (4) the \FSGS~algorithm proposed in \cite{feldman2020simultaneous}, abbreviated as ``FSGS''; and (5) our \ARG~algorithm, abbreviated as ``ARG''. Note that the three baseline algorithms REPG, TGF and FSGS achieve the best-known performance bounds among the related studies, as illustrated in Table~\ref{tab:overview}. In all experiments, we adopt the optimal settings of each implemented algorithm such that their theoretical approximation ratio is minimized (e.g., setting $\ell=2,p=\frac{2}{1+\sqrt{k}}$ for RAMG), and we set $\epsilon=0.1$ whenever $\epsilon$ is an input parameter for the considered algorithms. The implemented algorithms are tested in three applications, as elaborated in the following.

\subsection{Movie Recommendation} \label{sec:movie}

This application is also considered in~\cite{mirzasoleiman2016fast,feldman2017greed,haba2020streaming}, where there are a set $\N$ of movies and each movie is labeled by several genres chosen from a predefined set $G$. The goal is to select a subset $S$ of movies from $\N$ to maximize the utility
\begin{eqnarray}
f(S)=\sum_{u\in \mathcal{N}}\sum_{v\in S}M_{u,v}-\sum_{u\in S}\sum_{v\in S}M_{u,v}, \label{eqn:targetfunctionmovie}
\end{eqnarray}
under the constraint that the number of movies in $S$ labeled by genre $g$ is no more than $m_g$ for all $g\in G$ and $|S|\leq m$, where $m_g:g\in G$ and $m$ are all predefined integers. Intuitively, by using $M_{u,v}$ to denote the ``similarity'' between movie $u$ and movie $v$, the first and second factors in Eqn.~\eqref{eqn:targetfunctionmovie} encourage the ``coverage'' and ``diversity'' of the movie set $S$, respectively. It is indicated in \cite{mirzasoleiman2016fast,feldman2017greed} that the function $f(\cdot)$ is submodular and the problem constraint is essentially a $k$-system constraint with $k=|G|$. In our experiments, we use the MovieLens dataset~\cite{haba2020streaming} containing 1793 movies, where each movie $u$ is associated with a $25$-dimensional feature vector $t_u$ calculated from user ratings. We set $M_{u,v}=e^{-\lambda \mathrm{dist}(t_u,t_v)}$ where $\mathrm{dist}(t_u,t_v)$ denotes the Euclidean distance between $t_u$ and $t_v$ and $\lambda$ is set to 0.2. There are three genres ``Adventure'', ``Animation'' and ``Fantasy'' in MovieLens, and we set $m_g=10$ for all genres. 

\begin{figure}[ht]
	\begin{center}
		\centerline{\includegraphics[scale=0.36]{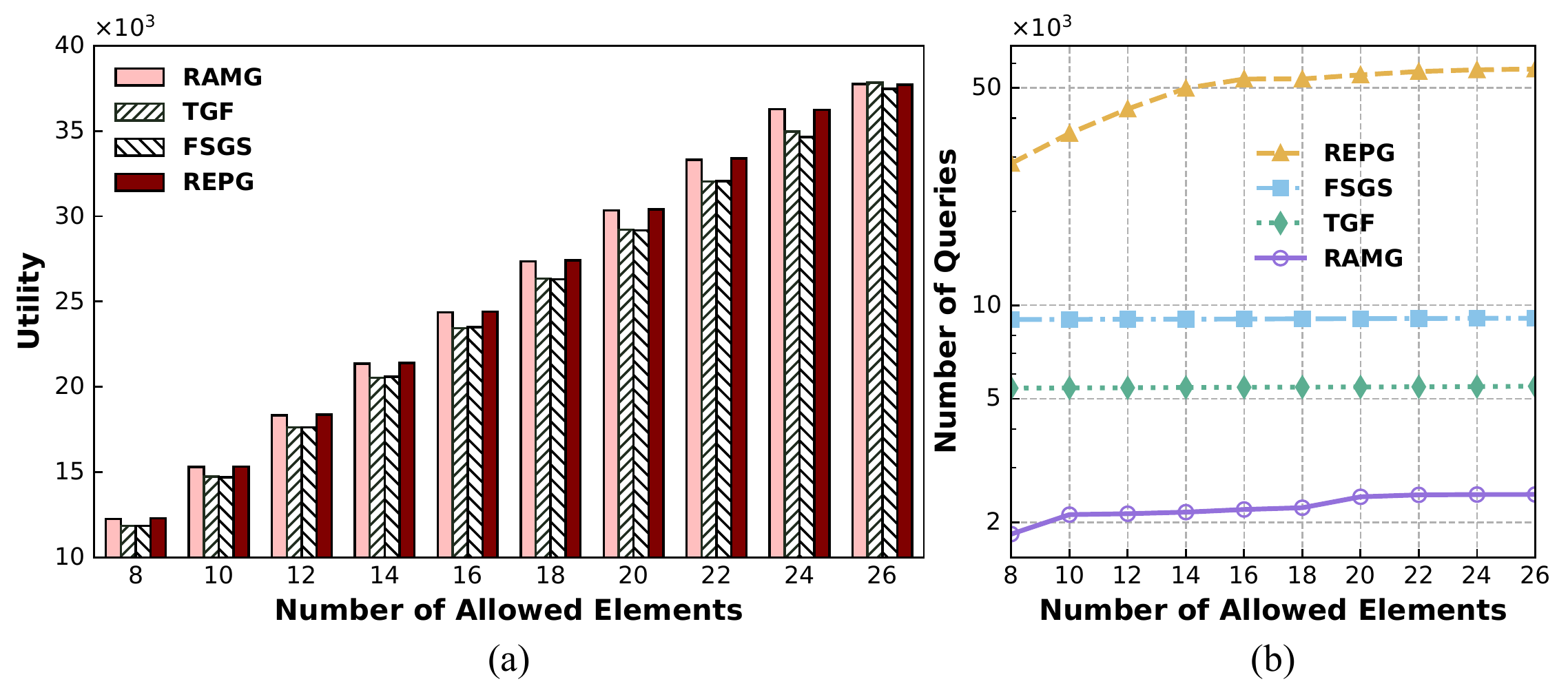}}
\vspace{-2ex}
		\caption{Movie Recommendation}
		\label{movie}
	\end{center}
\vspace{-3ex}
\end{figure}

In Fig.~\ref{movie}(a)-(b), we scale the the total number of movies allowed to be selected (i.e., $m$) to compare the performance of the implemented algorithms. It can be seen from Fig.~\ref{movie}(a) that RAMG and REPG achieve almost the same utility, while both of them outperform TGF and FSGS. Moreover, Fig.~\ref{movie}(b) shows that RAMG incurs much fewer oracle queries than all the baseline algorithms, and TGF is more efficient than FSGS. This can be explained by the reason that, FSGS maintains more candidate solutions than TGF, while the acceleration method adopted by RAMG is more efficient than the ``thresholding'' method adopted by TGF in practice.  



\subsection{Image Summarization} \label{sec:image}
This application is also considered in \cite{mirzasoleiman2016fast,fahrbach2019non}, where there is a set $\N$ of images classified into several categories, and the goal is to select a subset $S$ of images from $\N$ to maximize the utility
\vspace{-1.5ex}
\begin{eqnarray}
f(S)=\sum\nolimits_{u\in \mathcal{N}}\max\nolimits_{v\in S}s_{u,v}-\frac{1}{|\mathcal{N}|}\sum\nolimits_{u\in S}\sum\nolimits_{v\in S}s_{u,v} \nonumber
\end{eqnarray}
(where $s_{u,v}$ denotes the similarity between image $u$ and image $v$), under the constraint the the numbers of images in $S$ belonging to every category and the total number of images in $S$ are all bounded. It can be verified that such a constraint is a matroid (i.e., $1$-system) constraint. We perform the experiment using the CIFAR-10 dataset~\cite{krizhevsky2009learning} containing ten thousands $32\times 32$ color images. The similarity $s_{u,v}$ is computed
as the cosine similarity of the 3,072-dimensional pixel vectors of images $u$ and $v$. We restrict the selection of images from three categories: Airplane, Automobile and Bird, and the number of images selected from each category is bounded by 5.

In Fig.~\ref{image}, we plot the experimental results by scaling the number of images allowed to be selected. It can be seen from Fig.~\ref{image}(a) that RAMG and REPG achieve approximately the same utility and outperform FSGS and TGF again. Besides, TGF performs much worse than FSGS on utility in this application, as it uses a more rigorous stopping condition in its thresholding method and hence neglects many elements with small marginal gains. The results in Fig.~\ref{image}(b) show that the superiority of RAMG on efficiency still maintains, while REPG outperforms FSGS significantly. This can be explained by the fact that, the marginal gains of the unselected elements diminish vastly after a new element is selected in the image summarization application, so the performance of the thresholding method adopted in FSGS deteriorates to be close to a naive greedy algorithm, which results in its worse efficiency as FSGS maintains more candidate solutions than the other algorithms. In contrast, the performance of RAMG on efficiency is more robust against the variations of underlying data distribution.

\begin{figure}[t]
	\begin{center}
		\centerline{\includegraphics[scale=0.36]{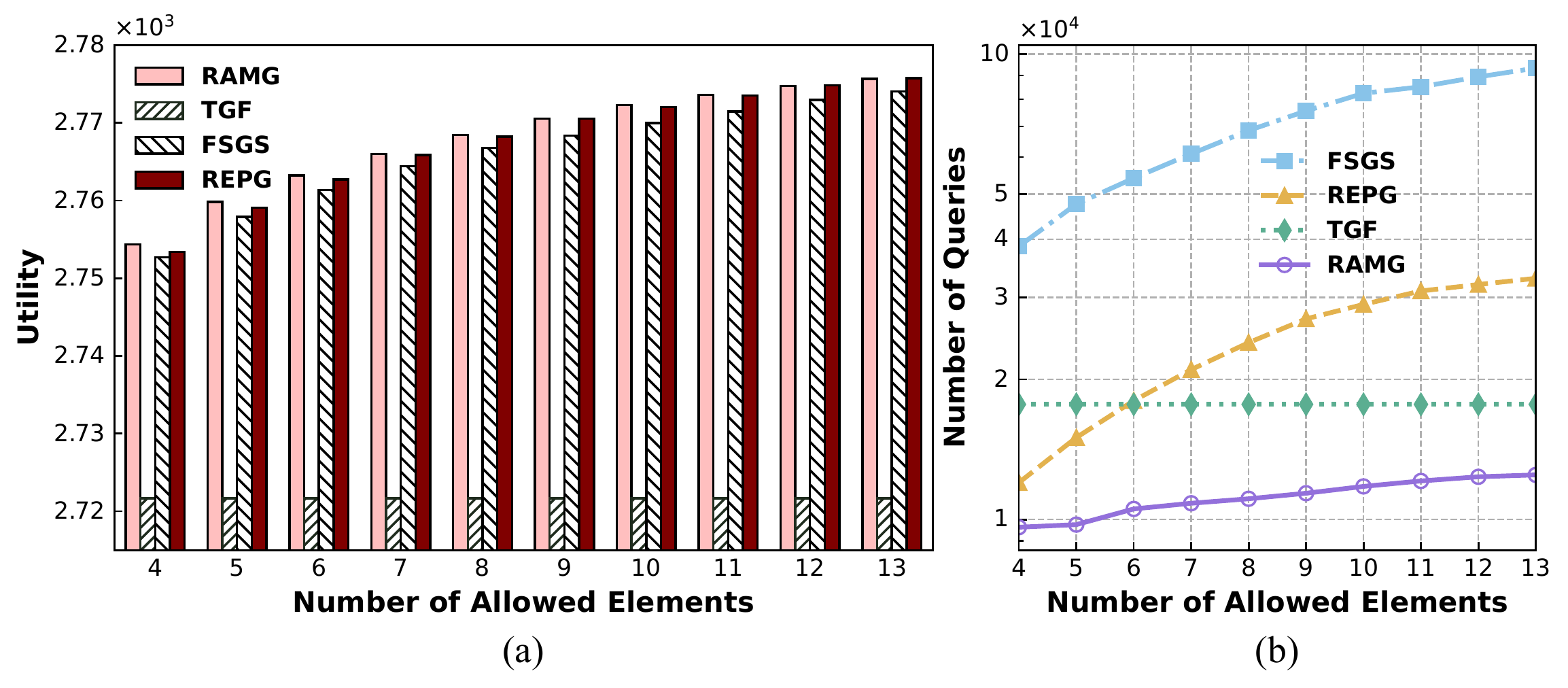}}
\vspace{-2ex}
		\caption{Image Summarization}
		\label{image}
	\end{center}
\vspace{-4ex}
\end{figure}



\subsection{Social Advertising with Multiple Products}

This application is also considered in~\cite{mirzasoleiman2016fast,fahrbach2019non,amanatidis2020fast}. We are given a social network $G=(\mathcal{N},E)$ where each node represents a user and each edge $(u,v)\in E$ is associated with a weight $w_{u,v}$ denoting the ``strength'' that $u$ can influence $v$. Suppose that there are $d$ kinds of products and an advertiser needs to select a ``seed'' set $H_i\subseteq \N$ for each $i\in [d]$, such that the total revenue can be maximized by presenting a free sample of product with type $i$ to each node in $H_i$. We also follow \cite{mirzasoleiman2016fast,amanatidis2020fast} to assume that the valuation of any user for a product is determined by the neighboring nodes owning the product with the same type, and the total revenue of $H_i$ is defined as
\begin{eqnarray}
f_i(H_i)=\sum\nolimits_{u\in \mathcal{N}\setminus H_i}\alpha_{u,i}\sqrt{\sum\nolimits_{v\in H_i}w_{v,u}},   
\end{eqnarray}
where $\alpha_{u,i}$ is a random number with known distributions. Suppose that each node $u\in \N$ can serve as a seed for at most $q$ types of products, and the total number of free samples available for any type of product is no more than $m$. The goal of the advertiser is to identify the seed sets $H_1,\cdots, H_d$ to maximize the expected value of $\sum_{i\in [d]}f_i(H_i)$ under the constraints described above. 
It is indicated in \cite{mirzasoleiman2016fast} that this problem is essentially a submodular maximization problem with a 2-system constraint.

We use the LastFM Social Network~\cite{barbieri2014influence,aslay2017revenue} with 1372 nodes and 14708 edges, and the edge weights in the network are randomly generated from the uniform distribution $\mathcal{U}(0,1)$. We adopt the same settings of \cite{amanatidis2020fast} to assume that, the parameter $\alpha_{u,i}$ follows a Pareto Type II distribution with $\lambda=1, \alpha = 2$ for all node $u$ and product $i$; and the parameters of $u$'s neighboring nodes can be observed after $u$ is selected under the adaptive setting. The values of $d$ and $q$ are set to $5$ and $3$, respectively. Following a comparison method in~\cite{amanatidis2020fast}, we also implement a variation of RAMG (dubbed RAMG$+$) where the input parameter $p$ is randomly sampled from (0.9,1). To test the performance of ARG, we randomly generate 20 realizations of the problem instance described above, and plot the average utility/number of queries of ARG on all the generated realizations.

\begin{figure}[ht]
	\begin{center}
		\centerline{\includegraphics[scale=0.36]{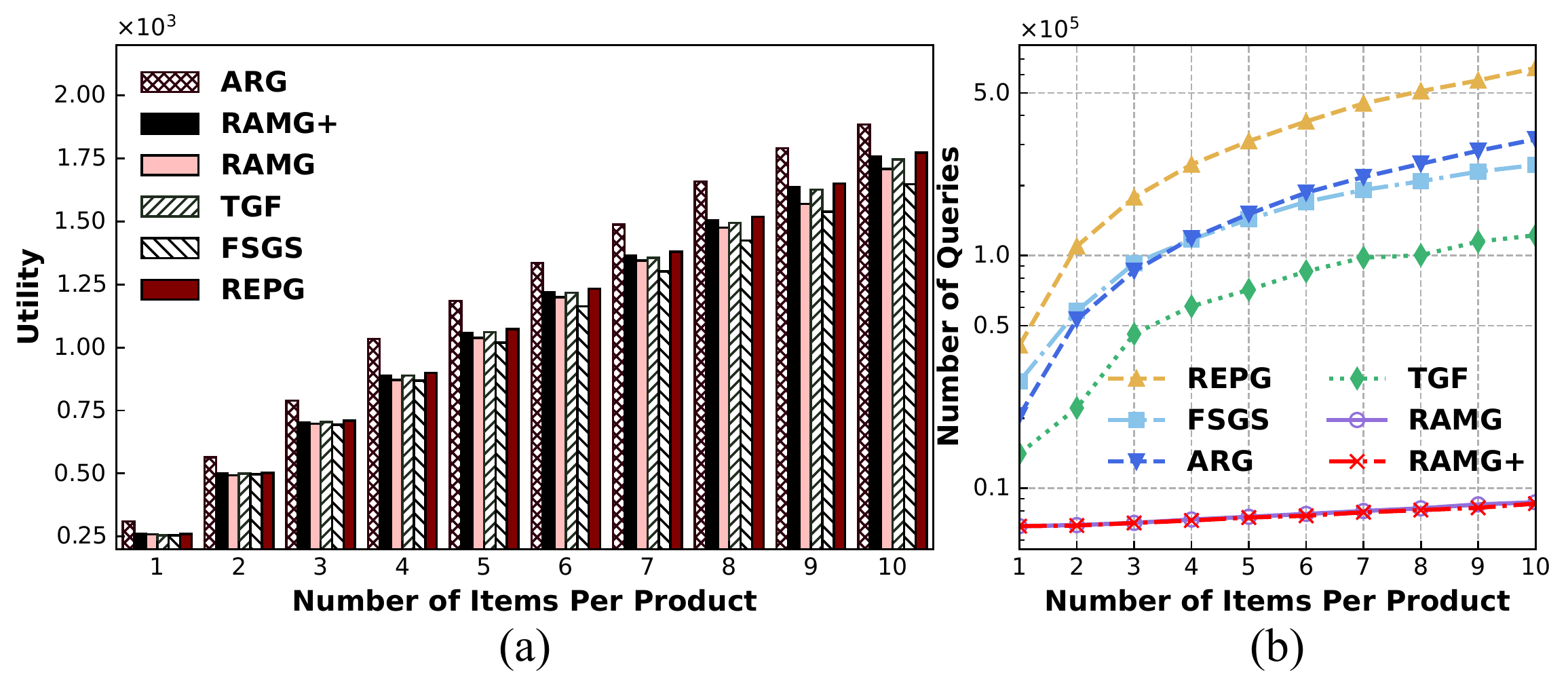}}
\vspace{-2ex}
		\caption{Social Advertising with Multiple Products}
		\label{social}
	\end{center}
\vspace{-3ex}
\end{figure}

We study the performance of all algorithms in Fig.~\ref{social} by scaling the number of items of each product available for seeding (i.e., $m$). It can be seen from Fig.~\ref{social}(a) that RAMG$+$, TGF and REPG achieve approximately the same utility, while the performance of RAMG and FSGS is slightly weaker. Note that RAMG$+$ has a weaker theoretical approximation ratio than RAMG. However, it is well known that approximation ratio is only a worst-case performance guarantee. Fig.~\ref{social}(a) also reveals that ARG performs the best on utility, which is not surprising as it can take advantage on side observation. In Fig.~\ref{social}(b), we compare the efficiency of all implemented algorithms and the results are qualitatively similar to those in Figs.~\ref{movie}-\ref{image}. Note that RAMG$+$ performs almost the same with RAMG on efficiency, which implies that it improves the utility performance of RAMG ``for free''.

\section{Conclusion} \label{sec:conclude}

We have proposed the first randomized algorithms for submodular maximization with a $k$-system constraint, under both the non-adaptive setting and the adaptive setting. Our algorithms outperform the existing algorithms in terms both approximation ratio and time complexity, and their superiorities have also been demonstrated by extensive experimental results on several applications related to data mining and social computing.

%

%


\bibliography{mylib}

\begin{thebibliography}{45}
\providecommand{\natexlab}[1]{#1}
\providecommand{\url}[1]{\texttt{#1}}
\expandafter\ifx\csname urlstyle\endcsname\relax
  \providecommand{\doi}[1]{doi: #1}\else
  \providecommand{\doi}{doi: \begingroup \urlstyle{rm}\Url}\fi

\bibitem[Amanatidis et~al.(2020)Amanatidis, Fusco, Lazos, Leonardi, and
  Reiffenh{\"a}user]{amanatidis2020fast}
Amanatidis, G., Fusco, F., Lazos, P., Leonardi, S., and Reiffenh{\"a}user, R.
\newblock Fast adaptive non-monotone submodular maximization subject to a
  knapsack constraint.
\newblock In \emph{Advances in Neural Information Processing Systems
  (NeurIPS)}, 2020.

\bibitem[Aslay et~al.(2017)Aslay, Bonchi, Lakshmanan, and Lu]{aslay2017revenue}
Aslay, C., Bonchi, F., Lakshmanan, L.~V., and Lu, W.
\newblock Revenue maximization in incentivized social advertising.
\newblock \emph{Proceedings of the VLDB Endowment (PVLDB)}, 10\penalty0
  (11):\penalty0 1238--1249, 2017.

\bibitem[Badanidiyuru \& Vondr{\'a}k(2014)Badanidiyuru and
  Vondr{\'a}k]{badanidiyuru2014fast}
Badanidiyuru, A. and Vondr{\'a}k, J.
\newblock Fast algorithms for maximizing submodular functions.
\newblock In \emph{ACM-SIAM Symposium on Discrete Algorithms (SODA)}, pp.\
  1497--1514, 2014.

\bibitem[Badanidiyuru et~al.(2014)Badanidiyuru, Mirzasoleiman, Karbasi, and
  Krause]{badanidiyuru2014streaming}
Badanidiyuru, A., Mirzasoleiman, B., Karbasi, A., and Krause, A.
\newblock Streaming submodular maximization: Massive data summarization on the
  fly.
\newblock In \emph{ACM SIGKDD International Conference on Knowledge Discovery
  and Data Mining (KDD)}, pp.\  671–680, 2014.

\bibitem[Badanidiyuru et~al.(2016)Badanidiyuru, Papadimitriou, Rubinstein,
  Seeman, and Singer]{badanidiyuru2016locally}
Badanidiyuru, A., Papadimitriou, C., Rubinstein, A., Seeman, L., and Singer, Y.
\newblock Locally adaptive optimization: Adaptive seeding for monotone
  submodular functions.
\newblock In \emph{ACM-SIAM Symposium on Discrete Algorithms (SODA)}, pp.\
  414--429, 2016.

\bibitem[Balkanski \& Singer(2018)Balkanski and Singer]{BalkanskiS18}
Balkanski, E. and Singer, Y.
\newblock The adaptive complexity of maximizing a submodular function.
\newblock In \emph{{STOC}}, pp.\  1138--1151, 2018.

\bibitem[Balkanski et~al.(2019)Balkanski, Rubinstein, and
  Singer]{balkanski2019optimal}
Balkanski, E., Rubinstein, A., and Singer, Y.
\newblock An optimal approximation for submodular maximization under a matroid
  constraint in the adaptive complexity model.
\newblock In \emph{ACM Symposium on Theory of Computing (STOC)}, pp.\  66--77,
  2019.

\bibitem[Barbieri \& Bonchi(2014)Barbieri and Bonchi]{barbieri2014influence}
Barbieri, N. and Bonchi, F.
\newblock Influence maximization with viral product design.
\newblock In \emph{Proceedings of the 2014 SIAM International Conference on
  Data Mining (SDM)}, pp.\  55--63, 2014.

\bibitem[Buchbinder et~al.(2014)Buchbinder, Feldman, Naor, and
  Schwartz]{buchbinder2014submodular}
Buchbinder, N., Feldman, M., Naor, J., and Schwartz, R.
\newblock Submodular maximization with cardinality constraints.
\newblock In \emph{ACM-SIAM Symposium on Discrete Algorithms (SODA)}, pp.\
  1433--1452, 2014.

\bibitem[Buchbinder et~al.(2015)Buchbinder, Feldman, Seffi, and
  Schwartz]{buchbinder2015tight}
Buchbinder, N., Feldman, M., Seffi, J., and Schwartz, R.
\newblock A tight linear time (1/2)-approximation for unconstrained submodular
  maximization.
\newblock \emph{SIAM Journal on Computing}, 44\penalty0 (5):\penalty0
  1384--1402, 2015.

\bibitem[Calinescu et~al.(2011)Calinescu, Chekuri, Pal, and
  Vondr{\'a}k]{calinescu2011maximizing}
Calinescu, G., Chekuri, C., Pal, M., and Vondr{\'a}k, J.
\newblock Maximizing a monotone submodular function subject to a matroid
  constraint.
\newblock \emph{SIAM Journal on Computing}, 40\penalty0 (6):\penalty0
  1740--1766, 2011.

\bibitem[Chekuri \& Quanrud(2019)Chekuri and Quanrud]{chekuri2019parallelizing}
Chekuri, C. and Quanrud, K.
\newblock Parallelizing greedy for submodular set function maximization in
  matroids and beyond.
\newblock In \emph{ACM Symposium on Theory of Computing (STOC)}, pp.\  78--89,
  2019.

\bibitem[Cuong \& Xu(2016)Cuong and Xu]{cuong2016adaptive}
Cuong, N.~V. and Xu, H.
\newblock Adaptive maximization of pointwise submodular functions with budget
  constraint.
\newblock In \emph{Advances in Neural Information Processing Systems
  (NeurIPS)}, pp.\  1252--1260, 2016.

\bibitem[Ene \& Nguyen(2019)Ene and Nguyen]{ene2019nearly}
Ene, A. and Nguyen, H.~L.
\newblock A nearly-linear time algorithm for submodular maximization with a
  knapsack constraint.
\newblock In \emph{International Colloquium on Automata, Languages and
  Programming (ICALP)}, pp.\  53:1--53:12, 2019.

\bibitem[Esfandiari et~al.(2021)Esfandiari, Amin, and
  Mirrokni]{Esfandiari2021adaptivity}
Esfandiari, H., Amin, K., and Mirrokni, V.
\newblock Adaptivity in adaptive submodularity.
\newblock In \emph{Conference on Learning Theory (COLT)}, 2021.

\bibitem[Fahrbach et~al.(2019{\natexlab{a}})Fahrbach, Mirrokni, and
  Zadimoghaddam]{fahrbach2019non}
Fahrbach, M., Mirrokni, V., and Zadimoghaddam, M.
\newblock Non-monotone submodular maximization with nearly optimal adaptivity
  and query complexity.
\newblock In \emph{International Conference on Machine Learning (ICML)}, pp.\
  1833--1842, 2019{\natexlab{a}}.

\bibitem[Fahrbach et~al.(2019{\natexlab{b}})Fahrbach, Mirrokni, and
  Zadimoghaddam]{Matthew2019}
Fahrbach, M., Mirrokni, V.~S., and Zadimoghaddam, M.
\newblock Submodular maximization with nearly optimal approximation, adaptivity
  and query complexity.
\newblock In \emph{{SODA}}, pp.\  255--273, 2019{\natexlab{b}}.

\bibitem[Feldman et~al.(2017)Feldman, Harshaw, and Karbasi]{feldman2017greed}
Feldman, M., Harshaw, C., and Karbasi, A.
\newblock Greed is good: Near-optimal submodular maximization via greedy
  optimization.
\newblock In \emph{Conference on Learning Theory (COLT)}, pp.\  758--784, 2017.

\bibitem[Feldman et~al.(2020)Feldman, Harshaw, and
  Karbasi]{feldman2020simultaneous}
Feldman, M., Harshaw, C., and Karbasi, A.
\newblock Simultaneous greedys: A swiss army knife for constrained submodular
  maximization.
\newblock \emph{arXiv:2009.13998}, 2020.

\bibitem[Fisher et~al.(1978)Fisher, Nemhauser, and Wolsey]{fisher1978analysis}
Fisher, M., Nemhauser, G., and Wolsey, L.
\newblock An analysis of approximations for maximizing submodular set
  functions—ii.
\newblock \emph{Mathematical Programming Study}, 8:\penalty0 73--87, 1978.

\bibitem[Fujii \& Sakaue(2019)Fujii and Sakaue]{fujii2019beyond}
Fujii, K. and Sakaue, S.
\newblock Beyond adaptive submodularity: Approximation guarantees of greedy
  policy with adaptive submodularity ratio.
\newblock In \emph{International Conference on Machine Learning (ICML)}, pp.\
  2042--2051, 2019.

\bibitem[Golovin \& Krause(2011{\natexlab{a}})Golovin and
  Krause]{golovin2011adaptive1}
Golovin, D. and Krause, A.
\newblock Adaptive submodularity: theory and applications in active learning
  and stochastic optimization.
\newblock \emph{Journal of Artificial Intelligence Research}, 42:\penalty0
  427--486, 2011{\natexlab{a}}.

\bibitem[Golovin \& Krause(2011{\natexlab{b}})Golovin and
  Krause]{golovin2011adaptive2}
Golovin, D. and Krause, A.
\newblock Adaptive submodular optimization under matroid constraints.
\newblock \emph{arXiv:1101.4450}, 2011{\natexlab{b}}.

\bibitem[Gomes \& Krause(2010)Gomes and Krause]{gomes2010budgeted}
Gomes, R. and Krause, A.
\newblock Budgeted nonparametric learning from data streams.
\newblock In \emph{International Conference on Machine Learning (ICML)}, pp.\
  391–398, 2010.

\bibitem[Gotovos et~al.(2015)Gotovos, Karbasi, and Krause]{gotovos2015non}
Gotovos, A., Karbasi, A., and Krause, A.
\newblock Non-monotone adaptive submodular maximization.
\newblock In \emph{International Joint Conference on Artificial Intelligence
  (IJCAI)}, pp.\  1996--2003, 2015.

\bibitem[Gupta et~al.(2010)Gupta, Roth, Schoenebeck, and
  Talwar]{gupta2010constrained}
Gupta, A., Roth, A., Schoenebeck, G., and Talwar, K.
\newblock Constrained non-monotone submodular maximization: Offline and
  secretary algorithms.
\newblock In \emph{International Workshop on Internet and Network Economics
  (WINE)}, pp.\  246--257, 2010.

\bibitem[Haba et~al.(2020)Haba, Kazemi, Feldman, and
  Karbasi]{haba2020streaming}
Haba, R., Kazemi, E., Feldman, M., and Karbasi, A.
\newblock Streaming submodular maximization under a $ k $-set system
  constraint.
\newblock In \emph{International Conference on Machine Learning (ICML)}, 2020.

\bibitem[Han et~al.(2018{\natexlab{a}})Han, Huang, and Luo]{HanTON2018}
Han, K., Huang, H., and Luo, J.
\newblock Quality-aware pricing for mobile crowdsensing.
\newblock \emph{IEEE/ACM Transactions on Networking}, 26\penalty0 (4):\penalty0
  1728--1741, 2018{\natexlab{a}}.

\bibitem[Han et~al.(2018{\natexlab{b}})Han, Huang, Xiao, Tang, Sun, and
  Tang]{HanPVLDB2018}
Han, K., Huang, K., Xiao, X., Tang, J., Sun, A., and Tang, X.
\newblock Efficient algorithms for adaptive influence maximization.
\newblock \emph{Proceedings of the {VLDB} Endowment}, 11\penalty0 (9):\penalty0
  1029--1040, 2018{\natexlab{b}}.

\bibitem[Han et~al.(2019)Han, Gui, Xiao, Tang, He, Cao, and
  Huang]{HanPVLDB2019}
Han, K., Gui, F., Xiao, X., Tang, J., He, Y., Cao, Z., and Huang, H.
\newblock Efficient and effective algorithms for clustering uncertain graphs.
\newblock \emph{Proceedings of the {VLDB} Endowment}, 12\penalty0 (6):\penalty0
  667--680, 2019.

\bibitem[Han et~al.(2020)Han, Cao, Cui, and Wu]{han2020deterministic}
Han, K., Cao, Z., Cui, S., and Wu, B.
\newblock Deterministic approximation for submodular maximization over a
  matroid in nearly linear time.
\newblock In \emph{Advances in Neural Information Processing Systems
  (NeurIPS)}, 2020.

\bibitem[Han et~al.(2021)Han, Cui, Zhu, Zhang, Wu, Yin, Xu, Tang, and
  Huang]{Han2021}
Han, K., Cui, S., Zhu, T., Zhang, E., Wu, B., Yin, Z., Xu, T., Tang, S., and
  Huang, H.
\newblock Approximation algorithms for submodular data summarization with a
  knapsack constraint.
\newblock \emph{Proceedings of the ACM on Measurement and Analysis of Computing
  Systems (POMACS)}, 5\penalty0 (1):\penalty0 05:1--05:31, 2021.

\bibitem[Iyer \& Bilmes(2013)Iyer and Bilmes]{iyer2013submodular}
Iyer, R.~K. and Bilmes, J.~A.
\newblock Submodular optimization with submodular cover and submodular knapsack
  constraints.
\newblock In \emph{Advances in Neural Information Processing Systems
  (NeurIPS)}, pp.\  2436--2444, 2013.

\bibitem[Karp et~al.(1988)Karp, Upfal, and Wigderson]{KarpUW88}
Karp, R.~M., Upfal, E., and Wigderson, A.
\newblock The complexity of parallel search.
\newblock \emph{Journal of Computer and System Sciences}, 36\penalty0
  (2):\penalty0 225--253, 1988.

\bibitem[Kempe et~al.(2003)Kempe, Kleinberg, and Tardos]{kempe2003maximizing}
Kempe, D., Kleinberg, J., and Tardos, {\'E}.
\newblock Maximizing the spread of influence through a social network.
\newblock In \emph{ACM SIGKDD International Conference on Knowledge Discovery
  and Data Mining (KDD)}, pp.\  137--146, 2003.

\bibitem[Krizhevsky et~al.(2009)Krizhevsky, Hinton,
  et~al.]{krizhevsky2009learning}
Krizhevsky, A., Hinton, G., et~al.
\newblock Learning multiple layers of features from tiny images.
\newblock 2009.

\bibitem[Kuhnle(2019)]{kuhnle2019interlaced}
Kuhnle, A.
\newblock Interlaced greedy algorithm for maximization of submodular functions
  in nearly linear time.
\newblock In \emph{Advances in Neural Information Processing Systems
  (NeurIPS)}, pp.\  2371--2381, 2019.

\bibitem[Lee et~al.(2010)Lee, Mirrokni, Nagarajan, and
  Sviridenko]{lee2010maximizing}
Lee, J., Mirrokni, V.~S., Nagarajan, V., and Sviridenko, M.
\newblock Maximizing nonmonotone submodular functions under matroid or knapsack
  constraints.
\newblock \emph{SIAM Journal on Discrete Mathematics}, 23\penalty0
  (4):\penalty0 2053--2078, 2010.

\bibitem[Mestre(2006)]{Mestre2006}
Mestre, J.
\newblock Greedy in approximation algorithms.
\newblock In \emph{European Symposium on Algorithms (ESA)}, pp.\  528--539,
  2006.

\bibitem[Minoux(1978)]{minoux1978accelerated}
Minoux, M.
\newblock Accelerated greedy algorithms for maximizing submodular set
  functions.
\newblock In \emph{Optimization techniques}, pp.\  234--243. Springer, 1978.

\bibitem[Mirzasoleiman et~al.(2016)Mirzasoleiman, Badanidiyuru, and
  Karbasi]{mirzasoleiman2016fast}
Mirzasoleiman, B., Badanidiyuru, A., and Karbasi, A.
\newblock Fast constrained submodular maximization: Personalized data
  summarization.
\newblock In \emph{International Conference on Machine Learning (ICML)}, pp.\
  1358--1367, 2016.

\bibitem[Mitrovic et~al.(2019)Mitrovic, Kazemi, Feldman, Krause, and
  Karbasi]{mitrovic2019adaptive}
Mitrovic, M., Kazemi, E., Feldman, M., Krause, A., and Karbasi, A.
\newblock Adaptive sequence submodularity.
\newblock In \emph{Advances in Neural Information Processing Systems
  (NeurIPS)}, pp.\  5353--5364, 2019.

\bibitem[Parthasarathy(2020)]{parthasarathy2020adaptive}
Parthasarathy, S.
\newblock Adaptive submodular maximization under acm symposium on theory of
  computing (stoc)hastic item costs.
\newblock In \emph{Conference on Learning Theory (COLT)}, pp.\  3133--3151,
  2020.

\bibitem[Quinzan et~al.(2021)Quinzan, Doskoc, G{\"{o}}bel, and
  Friedrich]{Francesco2021}
Quinzan, F., Doskoc, V., G{\"{o}}bel, A., and Friedrich, T.
\newblock Adaptive sampling for fast constrained maximization of submodular
  function.
\newblock In \emph{AISTATS}, 2021.

\bibitem[Singla et~al.(2016)Singla, Tschiatschek, and Krause]{singla2016noisy}
Singla, A., Tschiatschek, S., and Krause, A.
\newblock Noisy submodular maximization via adaptive sampling with applications
  to crowdsourced image collection summarization.
\newblock In \emph{AAAI Conference on Artificial Intelligence (AAAI)}, pp.\
  2037–2041, 2016.

\end{thebibliography}
\bibliographystyle{icml2021}

\appendix

\section{Missing Proofs from Section~\ref{sec:nonadaptive}}

\subsection{Proof of Lemma~\ref{lma:mapping}}

%
%
%
\red{
\begin{proof}
The proof is constructive and is inspired by~\cite{calinescu2011maximizing,han2020deterministic}. For clarity, we provide a procedure to construct $\sigma_i(\cdot)$, as shown by Algorithm~\ref{alg:mapping}. Suppose that the elements in $S_i$ are $\{z_1,\cdots, z_q\}$ (listed according to the order that they are added into $S_i$). Algorithm~\ref{alg:mapping} finds a series of sets $J_0\subseteq J_1\subseteq \cdots\subseteq J_q=Q_i$ such that all the elements in $M_t=J_t\setminus J_{t-1}$ is mapped to $z_t$ by $\sigma_i(\cdot)$ for any $t\in \{1,2,\cdots,q\}$. From Algorithm~\ref{alg:mapping}, it can be easily seen that $\sigma_i(\cdot)$ satisfies the conditions required by the lemma. The only problem left is to prove that all the elements in $Q_i$ is mapped by $\sigma_i(\cdot)$, i.e., to prove $J_0=\emptyset$. Indeed, we can prove a stronger result $\forall t\in \{0,1,\cdots,q\}: |J_t|\leq kt$ by induction:
%
%
	\begin{itemize}
		\item When $t=q$, we will prove $|J_q|\leq kq$ by showing that $S_i$ is a base of $Q_i\cup S_i$. It is obvious that each element $u\in  O_i^-$ satisfies $S_i\cup \{u\}\notin \mathcal{I}$ according to the definition of $O_i^-$. Moreover, for any element $u\in \cup_{j\in [\ell]\setminus\{i\}}(O_j^{i-}\cup \widehat{O}_j^{i-})$, we must have $S_i\cup \{u\}\notin \mathcal{I}$, because otherwise we have $S_i^<(u)\cup \{u\}\in \mathcal{I}$ due to $S_i^<(u)\subseteq S_i$ and the down-closed property of independence systems, contradicting the definition of $O_j^{i-}$ and $\widehat{O}_j^{i-}$. These reasoning implies that $S_i$ is a base of $Q_i\cup S_i$. Note that $Q_i\subseteq O$. So we can get $|J_q|=|Q_i|\leq k|S_i|= kq$ according to the definition of $k$-systems.
		\item Suppose that $|J_t|\leq kt$ holds, we will prove $|J_{t-1}|\leq k(t-1)$. If the set $C_t$ determined in Line~\ref{ln:defiofct} of Algorithm~\ref{alg:mapping} has a cardinality larger than $k$, then we have $|M_t|=k$ according to Algorithm~\ref{alg:mapping} and hence $|J_{t-1}|=|J_{t}|-k\leq k(t-1)$. If $|C_t|\leq k$, then $\{{z_1},\cdots,z_{t-1}\}$ must be a base of $\{{z_1},\cdots,z_{t-1}\}\cup J_{t-1}$, because there does not exist $u\in J_{t-1}\setminus \{{z_1},\cdots,z_{t-1}\}$ such that $\{{z_1},\cdots,z_{t-1}\}\cup \{u\}\in \mathcal{I}$ according to Algorithm~\ref{alg:mapping}. So we also have $|J_{t-1}|\leq k(t-1)$ according to $J_{t-1}\in \mathcal{I}$ and the definition of $k$-systems.
	\end{itemize}
	From the above reasoning we know $J_0=\emptyset$. So the lemma follows.
\end{proof}
}

\begin{algorithm} [h]
	\caption{\textsc{Constructing the Mapping $\sigma_i(\cdot)$}}
	\label{alg:mapping}
	\begin{algorithmic}
		\STATE {\bfseries Initialize:} Denote the elements in $S_{i}$ as $\{z_1,\cdots,z_q\}$, where elements are listed according to the order that they are added into $S_i$; $J_q\gets Q_i$
	\end{algorithmic}
	\begin{algorithmic}[1]
		\FOR{$t=q$ {\bfseries to} $0$}
		\STATE $C_t\gets \{e\in J_t\backslash \{{z_1},\cdots,z_{t-1}\}:\{{z_1,\cdots,z_{t-1},e}\}\in \I\}$ \label{ln:defiofct}
		\IF{$|C_t|\leq k$}
		\STATE $M_t\gets C_t$
		\ENDIF
		\IF{$|C_t|> k$}
		\IF{$z_t\in C_t$}
		\STATE Find a subset $M_t\subseteq C_t$ satisfying $|M_t|=k$ and $z_t\in M_t$
		\ELSE
		\STATE Find a subset $M_t\subseteq C_t$ satisfying $|M_t|=k$
		\ENDIF
		\ENDIF
		\STATE Let $\sigma_i(z)=z_t$ for all $z\in M_t$; $J_{t-1}\gets J_t\setminus M_t$
		\ENDFOR
	\end{algorithmic}
\end{algorithm}

\subsection{Proof of Lemma~\ref{lma:mgupperbound}}

    \begin{proof}
     We first prove Eqn.~\eqref{eqn:toprove1}. According to the definitions of $O_j^{i+}$ and $\widehat{O}_j^{i+}$, any element $u\in O_j^{i+}\cup \widehat{O}_j^{i+}$ can also be added into $S_i$ without violating the feasibility of $\I$ when $u$ is inserted into $S_j$. Therefore, according to the greedy rule of \RMG~and the submodularity of $f(\cdot)$, we must have
        \begin{eqnarray}
           \forall u\in O_j^{i+}\cup \widehat{O}_j^{i+}:  f(u\mid S_i)\leq f(u\mid S_i^<(u)))\leq f(u\mid S_j^<(u)))=\delta(u)
        \end{eqnarray}

     Now we prove Eqn.~\eqref{eqn:toprove2}. Recall that $Q_i=\cup_{j\in [\ell]\setminus\{i\}}(O_j^{i-}\cup \widehat{O}_j^{i-})\cup (O\cap S_i)\cup O_i^-$. According to Lemma~\ref{lma:mapping}, any element $u\in O_j^{i-}\cup \widehat{O}_j^{i-} (j\neq i)$ can be added into $S_i^<(\pi_i(u))$ without violating the feasibility of $\I$. Moreover, $u$ must have not been considered by the algorithm at the moment that $\pi_i(u)$ is added into $S_i$, because otherwise we have $S_i^<(u)\subseteq S_i^<(\pi_i(u))$ and hence $S_i^<(u)\cup \{u\}\in \I$ due to the definition of independence systems, which contradicts the definitions of $O_j^{i-}$ and $\widehat{O}_j^{i-}$. Therefore, according to the greedy rule of  \RMG~and submodularity, we can get
     \begin{eqnarray}
           \forall u\in O_j^{i-}\cup \widehat{O}_j^{i-}:  f(u\mid S_i)\leq f(u\mid S_i^<(u)))\leq f(u\mid S_i^<(\pi_i(u)))\leq f(\pi_i(u) \mid S_i^<(\pi_i(u)))=\delta(\pi_{i}(u))
     \end{eqnarray}
     By similar reasoning, we can also prove $\forall u\in O_i^{-}: f(u\mid S_i)\leq f(u\mid S_i^<(\pi_i(u)))\leq \delta(\pi_{i}(u))$. Finally, $f(u\mid S_i)\leq \delta(\pi_{i}(u))$ trivially holds for all $u\in O\cap S_i$ as $\pi_i(u)=u$ due to Lemma~\ref{lma:mapping}. So the lemma follows.
    \end{proof}

\subsection{Proof of Lemma~\ref{lma:boundsumfosi}}

As the proof of Lemma~\ref{lma:boundsumfosi} is a bit involved, we first introduce Lemma~\ref{lma:involved}, and then use Lemma~\ref{lma:involved} to prove Lemma~\ref{lma:boundsumfosi}.

\begin{lemma} \label{lma:involved}
We have
\begin{eqnarray}
\sum_{i\in[\ell]}\Bigg(\sum_{j\in[\ell]\setminus\{i\}}\bigg(\sum_{u\in O_{j}^{i+}}\delta(u)+\sum_{u\in O_{j}^{i-}\cup \widehat{O}_j^{i-}}\delta(\pi_i(u)) \bigg) + \sum_{u\in O_i^-}\delta(\pi_i(u))\Bigg)\leq \ell(k + \ell -2) f(S^*)
\end{eqnarray}
\end{lemma}

\begin{proof}

For any $i\in [\ell]$, let $\lambda(i)=(i\mod \ell)+1$. So we have


\begin{eqnarray}
\sum_{i\in[\ell]}\sum_{u\in O_{\lambda(i)}^{i+}}\delta(u)=\sum_{j\in[\ell]}\sum_{u\in O_{j}^{\lambda^{-1}(j)+}}\delta(u)\leq \sum_{j\in[\ell]}\sum_{u\in O\cap S_j}\delta(u)=\sum_{i\in[\ell]}\sum_{u\in O\cap S_i}\delta(u), \label{eqn:ilambdai}
\end{eqnarray}
where the inequality is due to $O_{j}^{\lambda^{-1}(j)+}\subseteq O\cap S_j$ and $\forall u\in S_j: \delta(u)>0$. So we can get 

\begin{eqnarray}
		\sum_{i\in[\ell]}\sum_{j\in[\ell]\setminus\{i\}}\sum_{u\in O_{j}^{i+}}\delta(u)&=&\sum_{i\in[\ell]}\bigg(\sum_{j\in[\ell]\setminus\{i,\lambda(i)\}}\sum_{u\in O_{j}^{i+}}\delta(u)+\sum_{u\in O_{\lambda(i)}^{i+}}\delta(u)\bigg) \nonumber\\
        &\leq&\sum_{i\in[\ell]}\sum_{j\in[\ell]\setminus\{i,\lambda(i)\}}\sum_{u\in O\cap S_j}\delta(u)+\sum_{i\in[\ell]}\sum_{u\in O\cap S_i}\delta(u)\label{eqn:touseilambadai}\\
        &\leq& \ell(\ell-2)f(S^*) + \sum_{i\in[\ell]}\sum_{u\in O\cap S_i}\delta(u) \label{eqn:moresimple}
\end{eqnarray}
where we leverage Eqn.~\eqref{eqn:ilambdai} to derive Eqn.~\eqref{eqn:touseilambadai}, and Eqn.~\eqref{eqn:moresimple} is due to $\sum_{u\in O\cap S_j}\delta(u)\leq \sum_{u\in S_j}\delta(u)\leq f(S_j)\leq f(S^*)$. Moreover, we can get
\begin{eqnarray}
&&\sum_{i\in[\ell]}\Bigg(\sum_{j\in[\ell]\setminus\{i\}}\bigg(\sum_{u\in O_{j}^{i+}}\delta(u)+\sum_{u\in O_{j}^{i-}\cup \widehat{O}_j^{i-}}\delta(\pi_i(u)) \bigg) + \sum_{u\in O_i^-}\delta(\pi_i(u))\Bigg) \nonumber\\
&=& \sum_{i\in[\ell]}\sum_{j\in[\ell]\setminus\{i\}}\sum_{u\in O_{j}^{i+}}\delta(u)+ \sum_{i\in[\ell]}\bigg(\sum_{j\in[\ell]\setminus\{i\}}\sum_{u\in O_{j}^{i-}\cup \widehat{O}_j^{i-}}\delta(\pi_i(u))+ \sum_{u\in O_i^-}\delta(\pi_i(u))\bigg) \nonumber\\
&\leq& \ell(\ell-2)f(S^*)+\sum_{i\in[\ell]}\bigg(\sum_{u\in O\cap S_i}\delta(u)+\sum_{j\in[\ell]\setminus\{i\}}\sum_{u\in O_{j}^{i-}\cup \widehat{O}_j^{i-}}\delta(\pi_i(u))+ \sum_{u\in O_i^-}\delta(\pi_i(u))\bigg) \label{eqn:touseabove}\\
&=& \ell(\ell-2)f(S^*)+\sum_{i\in[\ell]}\sum_{u\in Q_i} \delta(\pi_i(u)) \nonumber\\
&\leq& \ell(\ell-2)f(S^*)+ k \sum_{i\in[\ell]}\sum_{u\in S_i} \delta(u) \label{eqn:touselemmafirst}\\
&\leq& \ell(\ell-2)f(S^*)+ k \sum_{i\in[\ell]}f(S_i) \leq \ell(k+\ell-2) f(S^*) \label{eqn:touselemmasecond}
\end{eqnarray}
where $Q_i=\cup_{j\in [\ell]\setminus\{i\}}(O_j^{i-}\cup \widehat{O}_j^{i-})\cup (O\cap S_i)\cup O_i^-$ is defined in Lemma~\ref{lma:mapping}; Eqn.~\eqref{eqn:touseabove} is due to Eqn.~\eqref{eqn:moresimple}; and Eqn.~\eqref{eqn:touselemmafirst} is due to Lemma~\ref{lma:mapping}. So the lemma follows.
\end{proof}

Now we provide the proof of Lemma~\ref{lma:boundsumfosi}:

\begin{proof}
Let $G_i=[\cup_{i\in[\ell]\backslash \{i\}}(O_{j}^{i+}\cup O_{j}^{i-}\cup \widehat{O}_j^{i+}\cup \widehat{O}_j^{i-})]\cup O_i^-\cup [O\cap D_i]$ for all $i\in [\ell]$. It is not hard to see that $G_i\subseteq O\setminus S_i$ and $\forall u\in O\setminus (S_i\cup G_i): f(u\mid S_i)\leq 0$. Therefore, we can get
%
%
	\begin{eqnarray}
		&&\sum_{i\in[\ell]}\Big(f(O\cup S_i)-f(S_i)\Big) \nonumber
		\\
		&\leq&\sum_{i\in[\ell]}\Bigg(\sum_{j\in[\ell]\setminus\{i\}}\bigg(\sum_{u\in O_{j}^{i+}\cup \widehat{O}_j^{i+}}f(u\mid S_i)+\sum_{u\in O_{j}^{i-}\cup \widehat{O}_j^{i-}}f(u\mid S_i)\bigg) + \sum_{u\in O_i^-}f(u\mid S_i)+\sum_{u\in O\cap D_i}f(u\mid S_i)\Bigg)~~\label{eqn:duetosubmodular}
		\\
		&\leq&\sum_{i\in[\ell]}\Bigg(\sum_{j\in[\ell]\setminus\{i\}}\bigg(\sum_{u\in O_{j}^{i+}\cup \widehat{O}_j^{i+}}\delta(u)+\sum_{u\in O_{j}^{i-}\cup \widehat{O}_j^{i-}}\delta(\pi_i(u)) \bigg) + \sum_{u\in O_i^-}\delta(\pi_i(u))+\sum_{u\in O\cap D_i}\delta(u)\Bigg) \label{eqn:uselmaandsub}
		\\
		&=&\sum_{i\in[\ell]}\Bigg(\sum_{j\in[\ell]\setminus\{i\}}\bigg(\sum_{u\in O_{j}^{i+}}\delta(u)+\sum_{u\in O_{j}^{i-}\cup \widehat{O}_j^{i-}}\delta(\pi_i(u)) \bigg) + \sum_{u\in O_i^-}\delta(\pi_i(u))\Bigg) \nonumber\\
&&+ \sum_{i\in [\ell]}\bigg(\sum_{j\in [\ell]\setminus\{i\}}\sum_{u\in \widehat{O}_j^{i+}}\delta(u)+\sum_{u\in O\cap D_i}\delta(u)\bigg) \nonumber\\
 &\leq& \ell(k + \ell -2) f(S^*)+ \sum_{i\in [\ell]}\bigg(\sum_{j\in [\ell]\setminus\{i\}}\sum_{u\in \widehat{O}_j^{i+}}\delta(u)+\sum_{u\in O\cap D_i}\delta(u)\bigg) \label{eqn:touseintrolma}
	\end{eqnarray}
where Eqn.~\eqref{eqn:duetosubmodular} is due to submodularity of $f(\cdot)$; Eqn.~\eqref{eqn:uselmaandsub} is due to Lemma~\ref{lma:mgupperbound} and submodularity; and Eqn.~\eqref{eqn:touseintrolma} is due to Lemma~\ref{lma:involved}. Moreover, we can get
	\begin{eqnarray}
		&&\sum_{i\in [\ell]}\bigg(\sum_{j\in [\ell]\setminus\{i\}}\sum_{u\in \widehat{O}_j^{i+}}\delta(u)+\sum_{u\in O\cap D_i}\delta(u)\bigg)\leq \sum_{i\in [\ell]}\bigg(\sum_{j\in [\ell]\setminus\{i\}}\sum_{u\in O\cap D_j}\delta(u)+\sum_{u\in O\cap D_i}\delta(u)\bigg)  \label{eqn:moreisbigger}
		\\
		&=&\sum_{i\in [\ell]}\sum_{j\in [\ell]}\sum_{u\in O\cap D_j}\delta(u)= \ell\sum_{u\in \mathcal{N}}X_u\cdotp \delta(u) \label{eqn:tocombinefinish}
	\end{eqnarray}
where Eqn.~\eqref{eqn:moreisbigger} is due to $\widehat{O}_j^{i+}\subseteq O\cap D_j$ and $\forall u\in D_j: \delta(u)>0$. Combining Eqn.~\eqref{eqn:touseintrolma} and Eqn.~\eqref{eqn:tocombinefinish} finishes the proof of Lemma~\ref{lma:boundsumfosi}.
\end{proof}

\subsection{Proof of Lemma~\ref{lma:boundexptation}}

We first quote the following lemma presented in \cite{buchbinder2014submodular}:

\begin{lemma}
\cite{buchbinder2014submodular} Given a ground set $\N$ and any non-negative submodular function $g(\cdot)$ defined on $2^{\N}$, we have $\E[g(Y)]\geq (1-p)g(\emptyset)$ if $Y$ is a random subset of $\N$ such that each element in $\N$ appears in $Y$ with probability of at most $p$ (not necessarily independently).
\label{lma:randomplma}
\end{lemma}

With the above lemma, Lemma~\ref{lma:boundexptation} can be proved as follows:

\begin{proof}

We first prove Eqn.~\eqref{eqn:osiislarger}. Note that $S_1,S_2,\cdots, S_{\ell}$ are disjoint sets. Using submodularity, we have
\begin{eqnarray}
&&\sum_{i=1}^{\ell} f(S_i\cup O)\geq f(O)+f(S_1\cup S_2\cup O)+\sum_{i=3}^{\ell}f(S_i\cup O) \nonumber \\
&\geq& 2f(O)+f(S_1\cup S_2\cup S_3\cup O)+\sum_{i=4}^{\ell}f(S_i\cup O)\geq \cdots \geq (\ell-1)f(O) +f(\cup_{i=1}^{\ell}S_i\cup O) \label{eqn:totakeexpectation}
\end{eqnarray}
Let $g: 2^{\N}\mapsto \mathbb{R}_{\geq 0}$ be a non-negative submodular function defined as: $\forall S\subseteq \N: g(S)=f(S\cup O)$. As each element in $\N$ appears in $\cup_{i=1}^{\ell}S_i$ with probability of no more than $p$, We can use Lemma~\ref{lma:randomplma} to get
\begin{eqnarray}
\E[f(\cup_{i=1}^{\ell}S_i\cup O)]=\E[g(\cup_{i=1}^{\ell}S_i)]\geq (1-p)g(\emptyset)=(1-p)f(O) \label{eqn:oneminusp}
\end{eqnarray}
Combining Eqn.~\eqref{eqn:totakeexpectation} and Eqn.~\eqref{eqn:oneminusp} finishes the proof of Eqn.~\eqref{eqn:osiislarger}.


Next, we prove Eqn.~\eqref{eqn:xudeltaissmall}. For any $u\in \N$, let $Y_u=1$ if $u\in \cup_{i=1}^{\ell}S_i$ and $Y_u=0$ otherwise; let $\EV_u$ be an arbitrary event denoting all the random choices of \RMG~up until the time that $u$ is considered to be added into a candidate solution, or denoting all the randomness of \RMG~if $u$ is never considered. Note that we have $\sum_{u\in \N} Y_u\cdot \delta(u)\leq \sum_{i=1}^{\ell}f(S_i)$. Therefore, by the law of total probability, we only need to prove
\begin{eqnarray}
\forall u\in \N: \frac{1-p}{p}\E[Y_u\cdot\delta(u)\mid \EV_u]\geq \E[X_u\cdot\delta(u)\mid \EV_u] \label{eqn:xuyup1p}
\end{eqnarray}
for any event $\EV_u$ defined above. Note that we have $X_u=0$ and hence Eqn.~\eqref{eqn:xuyup1p} clearly holds if $u\notin O$ or $u$ is never considered by the algorithm. Otherwise we have $\E[Y_u\cdot\delta(u)\mid \EV_u]=p\cdot \delta(u)$ and $\E[X_u\cdot\delta(u)\mid \EV_u]=(1-p)\cdot\delta(u)$ due to the reason that $u$ is accepted with probability of $p$ and discarded with probability of $1-p$. Combining all these results completes the proof of Eqn.~\eqref{eqn:xudeltaissmall}.
\end{proof}

\begin{algorithm}[t]
    \caption{\textsc{Choose}($S_1,S_2,\cdots,S_{\ell}, v_1,\cdots, v_{\ell}, v^*$)}
    \label{alg:choose}
    \begin{algorithmic}[1]

        \IF {$\cup_{i=1}^{\ell}S_i=\emptyset$}   
            \STATE Let $A_i\gets \{u\in \N: \{u\}\in \I\wedge f(u\mid \emptyset)>0\}$ for all $i\in [\ell]$;
            \FORALL{$i\in [\ell]$\alglinelabel{algline:choose-init-start}}
                \STATE Let $w_i(u)\gets f(u\mid \emptyset)$ and $\tau_i(u)\gets 0$ for all $u\in A_i$;
                \STATE \text{Store $A_i$ as a priority list according to the non-increasing order of $w_i(u): u\in A_i$} for all $i\in [\ell]$;
                \STATE Let $v_i\gets \arg\max_{u\in A_i} w_i(u)$;
            \ENDFOR\alglinelabel{algline:choose-init-end}

        \ELSE

            \STATE $C\gets [\ell]\backslash \{j\in [\ell]: (v_j\neq v^*)\vee (v_j=\mathrm{NULL})\}$

            \FORALL{$i\in C$}

              \STATE Let $v_i\gets \mathrm{NULL}$ and remove all elements in $A_i$ with non-positive weights;

              \WHILE{$A_i\neq \emptyset$}

                \STATE pop out the top element $u$ from $A_i$;
                \IF{$f(u\mid S_i)$~\text{has been computed}}
                    \STATE $v_i\gets u;$~~\textbf{exit while};
                \ENDIF
                \IF{$S_i\cup \{u\}\notin\I$}
                    \STATE \textbf{continue};
                \ENDIF

                    \STATE $\mathit{old}\gets w_i(u);~\tau_i(u)\gets \tau_i(u)+1$;
                    \STATE Compute $f(u\mid S_i)$ and let $w_i(u)\gets f(u\mid S_i)$; 

                    \IF{$w_i(u)\geq\frac{\mathit{old}}{1+\epsilon}$}
                        \STATE $v_i\gets u;$~~\textbf{exit while};
                    \ELSE
                        \IF{$\tau_i(u)\leq \lceil \log_{1+\epsilon}\frac{\ell r}{\epsilon}\rceil$} \alglinelabel{ln:removedfromai}
                            \STATE re-insert $u$ into $A_i$ and resort the elements in $A_i$;
                        \ENDIF
                    \ENDIF

              \ENDWHILE
             \ENDFOR
          \ENDIF
        \STATE Let $i^*\gets \arg\max_{i\in [\ell]: v_i\neq \mathrm{NULL}}f(v_{i}\mid S_i)$ and remove $v_{i^*}$ from $A_i$ for all $i\in [\ell]$
        \STATE {\bfseries Output:} $v_1,v_2,\cdots, v_{\ell},i^*$
    \end{algorithmic}
\end{algorithm}

\vspace{-2ex}

\subsection{Proof of Theorem~\ref{thm:approxratioacc}}

For clarity, we first provide the detailed design of the accelerated version of \RMG, as shown by Algorithm~\ref{alg:accrmg}. In the $t$-th iteration, Algorithm~\ref{alg:accrmg} calls a procedure \textsc{Choose} to greedily find an candidate element $v_i$ for $S_i$ satisfying $f(v_i\mid S_i)>0$ and $S_i\cup \{v_i\}\in \I$ for each $i\in [\ell]$. The \textsc{Choose} procedure also returns an index $i_t$ same to that in Algorithm~\ref{alg:muti-randomgreedy}. After that, Algorithm~\ref{alg:accrmg} runs similarly as Algorithm~\ref{alg:muti-randomgreedy}, i.e., it inserts $v_{i_t}$ into $S_{i_t}$ with probability $p$, and then enters the $(t+1)$-th iteration. Note that the elements $v_1,\cdots,v_{\ell}$ and $v_{i_t}$ found in the $t$-th iteration are also used to call \textsc{Choose} in the $(t+1)$-th iteration, so that \textsc{Choose} need not to identify a new $v_i$ for all $i\in [\ell]: v_i\neq v_{i_t}$ (as $S_i$ does not change for these $i$'s) and hence time efficiency can be improved. Finally, Algorithm~\ref{alg:accrmg} returns the optimal set among $S_1,\cdots, S_{\ell}$ and $S_0$, where $S_0$ is the singleton set with the maximum utility.

Next, we provide a brief description on the \textsc{Choose} procedure. As explained in Sec.~\ref{sec:accel}, \textsc{Choose} maintains $\ell$ sets $A_1, A_2,\cdots, A_{\ell}$ such that $v_i$ can be selected from $A_i$. At the first time that \textsc{Choose} is called, \textsc{Choose} assigns each element $u\in A_i$ a weight $w_i(u)=f(u\mid \emptyset)$ and an integer $\tau_i(u)$ indicating how many times $w_i(u)$ has been updated \red{(Lines~\ref{algline:choose-init-start}--\ref{algline:choose-init-end})}. Afterwards, \textsc{Choose} runs as that described in Sec.~\ref{sec:accel} and finds $v_i$ for each $i\in [\ell]$. Finally, \textsc{Choose} identifies $v_{i^*}$ from $\{v_i: i\in [\ell]\}$ which has the maximum marginal gain, and it also removes $v_{i^*}$ from all $A_i:i\in [\ell]$ because $v_{i^*}$ will used as $v_{i_t}$ by Algorithm~\ref{alg:accrmg}.

Note that Algorithm~\ref{alg:accrmg} differs from Algorithm~\ref{alg:muti-randomgreedy} in two points: (1) the element $u_t$ found in the $t$-th iteration is only an $(\frac{1}{1+\epsilon})$-approximate solution; (2) there are elements removed from $A_i$ due to ``too many updates''. Based on this observation, we can slightly modify the proofs for Algorithm~\ref{alg:muti-randomgreedy} to prove Theorem~\ref{thm:approxratioacc}, as presented below:


\begin{proof}
Let $L_i$ denote the set of all elements removed from $A_i$ due to Line~\ref{ln:removedfromai} of Algorithm~\ref{alg:choose}. We can slightly modify Definition~\ref{def:keydef} to re-define the sets $O_j^{i+},O_j^{i-},\widehat{O}_j^{i+},\widehat{O}_j^{i-}, O_i^-$ as follows:
    \begin{eqnarray*}
        &&O_j^{i+}=\left\{u\in O\cap S_j : S_i^<(u)\cup \{u\} \in \mathcal I \right\}\setminus L_i; \\
        &&O_j^{i-}=\left\{u\in O\cap S_j : S_i^<(u)\cup \{u\} \notin \mathcal I \right\}\setminus L_i; \\
        &&\widehat{O}_j^{i+}=\left\{u\in O\cap D_j : S_i^<(u)\cup \{u\} \in \mathcal I \right\}\setminus L_i; \\
        &&\widehat{O}_j^{i-}=\left\{u\in O\cap D_j : S_i^<(u)\cup \{u\} \notin \mathcal I \right\}\setminus L_i;\\
        &&O_i^-=\left\{u\in O\setminus{U} : S_i\cup \{u\} \notin \mathcal I \wedge f(u\mid S_i)>0 \right\}\setminus L_i;
    \end{eqnarray*}
With this new definition, it can be easily verified that each element $u$ in $O_j^{i+}\cup O_j^{i-}\cup \widehat{O}_j^{i+}\cup \widehat{O}_j^{i-}$ is still a candidate considered for $S_i$ in the \textsc{Choose} procedure when the algorithm tries to insert $u$ into $S_j$. Therefore, according to the greedy rule of \RMG~and the $(1+\epsilon)^{-1}$-approximation ratio of \textsc{Choose}, we can use similar reasoning as that for Lemma~\ref{lma:mgupperbound} to prove
\begin{eqnarray}
        &&\forall u\in O_j^{i+}\cup \widehat{O}_j^{i+}: f(u\mid S_i)\leq (1+\epsilon)\delta(u);~~~\label{eqn:toprove1new} \\
        &&\forall u\in \cup_{j\in [\ell]\setminus\{i\}}(O_j^{i-}\cup \widehat{O}_j^{i-})\cup (O\cap S_i)\cup O_i^-: f(u\mid S_i)\leq (1+\epsilon)\delta(\pi_i(u));~~~\label{eqn:toprove2new}
\end{eqnarray}
With the above results, we can use similar reasoning as that in Lemma~\ref{lma:boundsumfosi} to prove:
\begin{eqnarray}
\frac{1}{1+\epsilon}\sum_{i\in[\ell]}f(O\mid S_i)\leq \ell(k + \ell -2) f(S^*)+\ell\sum_{u\in \N} X_u\cdot \delta(u)+\sum_{i\in [\ell]}\sum_{u\in L_i\cap O}f(u\mid S_i) \label{eqn:tocombine11}  
\end{eqnarray}
Moreover, we have
\begin{eqnarray}
\sum_{u\in L_i\cap O}f(u\mid S_i)\leq \sum_{u\in L_i\cap O}f(u\mid \emptyset)(1+\epsilon)^{-\lceil \log_{1+\epsilon}\frac{\ell r}{\epsilon} \rceil}\leq \sum_{u\in L_i\cap O}\frac{\epsilon}{\ell r}f(u)\leq \epsilon f(S^*)/\ell \label{eqn:tocombine12}
\end{eqnarray}
where the first inequality is due the reason that the weight of each element $u\in L_i$ have been updated in \textsc{Choose} procedure for more than $\lceil \log_{1+\epsilon}\frac{\ell r}{\epsilon} \rceil$ times and it diminishes by a factor of $\frac{1}{1+\epsilon}$ for each update. Combining Eqn.~\eqref{eqn:tocombine11}, Eqn.~\eqref{eqn:tocombine12} and Lemma~\ref{lma:boundexptation}, we can prove
\begin{eqnarray}
f(O)\leq \bigg[(1+\epsilon)\frac{\ell (k+\frac{\ell}{p}-1)}{\ell-p}-\frac{(\ell-1)\epsilon-\epsilon^2}{\ell-p}\bigg] \E[f(S^*)]
\end{eqnarray}
Therefore, the approximation ratio of the accelerated \textsc{RandomMultiGreedy}~algorithm is at most $(1+\epsilon)(1+\sqrt{k})^2$ when $\ell=2,p=\frac{2}{1+\sqrt{k}}$ (for a randomized algorithm), or at most $(1+\epsilon)(k+\sqrt{k}+\lceil \sqrt{k}\rceil+1)$ when $\ell=\lceil\sqrt{k}\rceil+1,p=1$ (for a deterministic algorithm). Finally, it can be seen that the \textsc{Choose} procedure incurs at most $\mathcal{O}(\log_{1+\epsilon}\frac{\ell r}{\epsilon})$ value and independence oracle queries for each element in each $A_i: i\in [\ell]$. So the total time complexity of the accelerated \textsc{RandomMultiGreedy}~algorithm is at most $\mathcal{O}(\ell n\log_{1+\epsilon}\frac{\ell r}{\epsilon})=\mathcal{O}(\frac{\ell n}{\epsilon}\log \frac{\ell r}{\epsilon})$, which completes the proof.
\end{proof}

\begin{algorithm}[t]
    \caption{\textsc{RandomMultiGreedy}$(\ell,p)$~~/*with acceleration*/}
    \label{alg:accrmg}
    \begin{algorithmic}
        \STATE {\bfseries Initialize:} $\forall i\in [\ell]:\ S_i\gets \emptyset; v_i\gets \mathrm{NULL};~~t\gets 1;u_0\gets \mathrm{NULL};$
    \end{algorithmic}
    \begin{algorithmic}[1]
        \REPEAT

        \STATE $(v_1,v_2,\cdots, v_{\ell},i_t)\gets$ \textsc{Choose}$(S_1,\cdots, S_{\ell}, v_1,\cdots, v_{\ell}, u_{t-1})$
        \IF{$\exists j\in [\ell]: v_j\neq \mathrm{NULL}$}
         \STATE $u_t\gets v_{i_t}$;
            \STATE $\mathbf{With}$ probability $p$ $\mathbf{do}$ $S_{i_t}\gets S_{i_t}\cup \{u_{t}\}$  
            \STATE $t\gets t+1$
        \ENDIF


        \UNTIL{$(\forall i\in [\ell]: v_i=\mathrm{NULL})$}
        \STATE $u^*\gets \arg\max_{u\in \N\wedge \{u\}\in \I}f(u);~S_0\gets \{u^*\}$
        \STATE $S^*\gets \arg\max_{S\in \{S_0,S_1,S_2,\cdots,S_\ell \}} f(S);~T\gets t-1$ \\
        \STATE {\bfseries Output:} $S^*, T$
    \end{algorithmic}
\end{algorithm}

%
%
%

\section{Missing Proofs from Section~\ref{sec:brg}}

\subsection{Proof of Lemma~\ref{lma:boundingmarginalgain}}
\begin{proof}
Consider the while-loop of Algorithm~\ref{alg:brg} in which $e_i$ is added into $S$. Let $\{a_1, a_2,\cdots,a_d\}$ and $\tau$ denote the random sequence returned by \textsc{RandSEQ} and the threshold considered in that iteration, respectively. Suppose that $e_i=a_q \in \{a_1, a_2,\cdots,a_d\}$. According to Line~\ref{ln:definejmini} of Algorithm~\ref{alg:brg}, we have
\begin{eqnarray}
\E[f(e_i\mid S_{i-1})\mid S_{i-1}]\geq \frac{|C_{q-1}|}{|C|}\cdot \tau\geq \tau/(1+\epsilon) \label{eqn:bigthantaueqn}
\end{eqnarray}
Suppose by contradiction that there exists $u\in \N\backslash U$ satisfying $S_{i-1}\cup \{u\}\in \I$ and $f(u\mid S_{i-1})> (1+\epsilon)^2 \E[f(e_i\mid S_{i-1})\mid S_{i-1}]$. Then we must have $f(u\mid S_{i-1})\geq (1+\epsilon)\tau$ according to Eqn.~\eqref{eqn:bigthantaueqn} and hence $\tau< \tau_{max}$. Using the down-closed property of independence systems and submodularity, this implies that $S^{(1+\epsilon)\tau}\cup \{u\}\in \I$ and $f(u\mid S^{(1+\epsilon)\tau})> (1+\epsilon)\tau$ . However, this contradicts the fact that, given any threshold $\tau$ considered by Algoirthm~\ref{alg:brg}, no elements with marginal gain larger than $\tau$ can be added into $S^{\tau}$ without violating $\I$ when the while-loop using $\tau$ in Algoirthm~\ref{alg:brg} is finished.
\end{proof}

\subsection{Proof of Theorem~\ref{thm:ratioofbrg}}
To prove the theorem, we first introduce the following lemma:

\begin{lemma}
We have $\E[f(O^-\mid S)]\leq (1+\epsilon)^2 k\cdot \E[f(S)]$
\label{lma:boundingominus}
\end{lemma}
\begin{proof}

For every possible pair of $(S,O^-)$, we can use similar reasoning as that in Lemma~\ref{lma:mapping} to find a mapping $\sigma: O^-\mapsto S$ such that at most $k$ elements are mapped to the same element in $S$ and $\{u\}\cup S^<(\sigma(u))\in \I$ for all $u\in O^-$. Recall that the elements in $S$ are denoted by $\{e_1,\cdots,e_r\}$. For any $i\in [r]$, let $\sigma^{-1}(e_i)$ denote the set of elements in $O^-\backslash S$ that are mapped to $e_i$ by $\sigma(\cdot)$ and let $\sigma^{-1}(e_i)=\emptyset$ if $e_i$ is a dummy element. So we have $|\sigma^{-1}(e_i)|\leq k$ for all $i\in [r]$. Using Lemma~\ref{lma:boundingmarginalgain} and submodularity, we get
\begin{eqnarray}
&&\E[f(O^-\mid S)]\leq \E\left[\sum_{u\in O^-\backslash S} f(u\mid S)\right]\leq \E\left[\sum_{u\in O^-\backslash S} f(u\mid S^<(\sigma(u)))\right]\\
&\leq& \E\left[\sum_{i\in [r]}\sum_{u\in \sigma^{-1}(e_i)} f(u\mid S_{i-1})\right]\leq k\cdot\E\left[\sum_{i\in [r]}\max_{u\in \sigma^{-1}(e_i)} f(u\mid S_{i-1})\right]\\
&\leq& (1+\epsilon)^2 k\cdot \E\left[ \sum_{i\in [r]} \E\left[f(e_i\mid S_{i-1})\mid S_{i-1}\right] \right]\\
&=& (1+\epsilon)^2 k\cdot \sum_{i\in [r]} \E\left[ \E\left[f(e_i\mid S_{i-1})\mid S_{i-1}\right] \right]\\
&=& (1+\epsilon)^2 k\cdot \sum_{i\in [r]} \E\left[f(e_i\mid S_{i-1})\right]=(1+\epsilon)^2 k\cdot \E\left[\sum_{i\in [r]} f(e_i\mid S_{i-1})\right]\\
&\leq&(1+\epsilon)^2 k\cdot \E\left[f(S)\right]
\end{eqnarray}
So the lemma follows.
\end{proof}

Now we use Lemma~\ref{lma:boundingominus} to prove Theorem~\ref{thm:ratioofbrg}:

\begin{proof}
For any $u\in \N$, define $X_u=1$ if $u\in (U\cap O)\backslash S$, otherwise define $X_u=0$. Note that any element $u\in O\backslash S$ must satisfy one of the following conditions: (1) $u\in O^-$; (2) $u\in (U\cap O)\backslash S$; (3) $f(u\mid S)\leq \tau_{min}$. So we can use submodularity, Lemma~\ref{lma:boundingominus} and $\tau_{max}\leq f(S)$ to get
\begin{eqnarray}
&&\E[f(O\mid S)]\leq r\tau_{min}+ \E[f(O^-\mid S)]+ \E\left[\sum_{u\in \N} X_u\cdot f(u\mid S^<(u))\right]\\
&\leq& \epsilon\cdot \E[f(S)]+ (1+\epsilon)^2 k\cdot \E[f(S)]+ \E\left[\sum_{u\in \N} X_u\cdot f(u\mid S^<(u))\right]
\end{eqnarray}
Meanwhile, as each element in $\N$ appears in $S$ with probability of no more than $p$, we can use similar reasoning as that in the proof of Lemma~\ref{lma:boundexptation} to get
\begin{eqnarray}
&&\E[f(S\cup O)]\geq (1-p)f(O);\\
&&\E\left[\sum_{u\in \N} X_u\cdot f(u\mid S^<(u))\right]\leq \frac{1-p}{p}\E[f(S)]
\end{eqnarray}
Combining the above equations completes the proof.
\end{proof}

\section{Missing Proofs from Section~\ref{sec:adaptopt}}

\subsection{Proof of Lemma~\ref{lma:concat}}

%

\begin{proof}
	Given any element set $Y\subseteq \N$ and any realization $\phi$, let $g(Y,\phi):=f(Y\cup \N(\pi_{\opt},\phi),\phi)$. \red{It is easy to verify that the non-negative function $g(\cdot, \phi)$ is  submodular.} Thus, given a fixed realization $\phi$, by Lemma~\ref{lma:randomplma}, we know that
	\begin{equation}
		\E_{\pi_{\A}}[g(\N(\pi_{\A},\phi),\phi)]\geq (1-p)g(\emptyset,\phi)
	\end{equation}
	Therefore, we have
	\begin{equation}
	f_{\avg}(\pi_{\opt}@\pi_{\mathcal{A}})=\E_{\Phi}[\E_{\pi_{\A}}[g(\N(\pi_{\A},\Phi),\Phi)]]\geq\E_{\Phi}[(1-p)g(\emptyset,\Phi)]=(1-p)f_{\avg}(\pi_{\opt}),
	\end{equation}
	which completes the proof.
\end{proof}

%
%
%

\subsection{Proof of Lemma~\ref{lma:upperboundofmg}}

\begin{proof}
	We first give an equivalent expression of the expected utility by a function of conditional \red{expected marginal gains}. Given a deterministic policy $\pi$ and a realization $\phi$, for each $u\in \N$, let $Y_{u}(\phi)$ be a boolean random variable such that $Y_{u}(\phi)=1$ if $u\in \N(\pi,\phi)$ and $Y_{u}(\phi)=0$ otherwise. Further, denote by $\psi_u^\pi(\phi)$ the partial realization observed by $\pi$ right before considering $u$ under realization $\phi$, and denote by $\Psi_u^\pi$ a random partial realization right before considering $u$ by $\pi$. We also use $Y_{u}(\psi_u^\pi(\phi))$ to represent $Y_{u}(\phi)$, since the partial realization $\psi_u^\pi(\phi)$ suffices to determine whether $u$ is added to the solution under realization $\phi$. Thus,
	\begin{eqnarray}
	&&\E_{\Phi} [f(\N(\pi,\Phi),\Phi)] \nonumber\\
	&=&\E_{\Phi} \Big[\sum_{u\in \N}\Big(Y_{u}(\Phi)\cdot\big(f(\dom(\psi_u^\pi(\Phi))\cup\{u\},\Phi)-f(\dom(\psi_u^\pi(\Phi)),\Phi)\big)\Big)\Big] \nonumber\\
	&=&\sum_{u\in \N}\E_{\Psi_u^\pi}\Big[\E_{\Phi} \Big[Y_{u}(\Phi)\cdot\big(f(\dom(\Psi_u^\pi)\cup\{u\},\Phi)-f(\dom(\Psi_u^\pi),\Phi)\big)\bigm| \Phi\sim \Psi_u^\pi\Big]\Big] \nonumber\\
	&=&\sum_{u\in \N}\E_{\Psi_u^\pi} \Big[Y_{u}(\Psi_u^\pi)\cdot\Delta(u\mid \Psi_u^\pi)\Big] =\sum_{u\in \N}\E_{\Phi}\Big[\E_{\Psi_u^\pi} \Big[Y_{u}(\Psi_u^\pi)\cdot\Delta(u\mid \Psi_u^\pi)\bigm|  \Phi\sim \Psi_u^\pi\Big]\Big] \nonumber\\
	&=&\sum_{u\in \N}\E_{\Phi}\Big[Y_{u}(\Phi)\cdot\Delta(u\mid \psi_u^\pi(\Phi))\Big]=\E_{\Phi} \Big[\sum_{u\in \N(\pi,\Phi)}\Delta(u\mid \psi_u^\pi(\Phi))\Big].
	\end{eqnarray}
	
	Denote by $\psi(\pi_{\A},\phi)$ the observed partial realization at the end of $\pi_{\A}$ under realization $\phi$. Then, similar to the above analysis, we have
	\begin{align*}
		f_{\avg}(\pi_{\A}@\pi_{\opt})
		&=\E_{\Phi,\pi_{\A}@\pi_{\opt}}[f(\N(\pi_{\A}@\pi_{\opt},\Phi),\Phi)]\\
		&=\E_{\pi_{\A}@\pi_{\opt}} \Big[\sum_{u\in \N(\pi_{\A},\Phi)}\Delta(u\mid \psi_u(\Phi))+\sum_{u\in \N(\pi_{\opt},\Phi)\setminus\N(\pi_{\A},\Phi)}\Delta(u\mid \psi(\pi_{\A},\Phi)\cup\psi_u^{\pi_{\opt}}(\Phi))\Big]\\
		&=f_{\avg}(\pi_{\A})+ \E_{\pi_{\A}@\pi_{\opt}} \Big[\sum_{u\in \N(\pi_{\opt},\Phi)\setminus\N(\pi_{\A},\Phi)}\Delta(u\mid \psi(\pi_{\A},\Phi)\cup\psi_u^{\pi_{\opt}}(\Phi))\Big]\\
		&\leq f_{\avg}(\pi_{\A})+ \E_{\pi_{\A}} \Big[\sum_{u\in \N(\pi_{\opt},\Phi)\setminus\N(\pi_{\A},\Phi)}\Delta(u\mid \psi_u(\Phi))\Big],
	\end{align*}
    where the inequality is due to adaptive submodularity and $\psi_u(\Phi)\subseteq\psi(\pi_{\A},\Phi)\subseteq \psi(\pi_{\A},\Phi)\cup\psi_u^{\pi_{\opt}}(\Phi)$.
\end{proof}

\subsection{Proof of Lemma~\ref{lma:boundo1}}

\begin{proof}
    Since $f_{\avg}(\pi_{\A})= \E_{\pi_{\A}}\Big[\E_{\Phi} \Big[\sum_{u\in \N(\pi_{\A},\Phi)}\Delta(u\mid \psi_{u}(\Phi))\Big]\Big]$, it suffices to prove
    \begin{equation}
        \sum\nolimits_{u\in O_1(\phi)} \Delta(u\mid \psi_{u}(\phi))\leq k\cdot \sum\nolimits_{u\in \N(\pi_{\A},\phi)}\Delta(u\mid \psi_{u}(\phi)) \label{eqn:mustholdkeyeqn}
    \end{equation}
%
    for any given realization $\phi\in Z^{\N}$ and fixed randomness of $\pi_{\A}$. Given a realization $\phi$, let $\hat{u}_i$ be the $i$-th element selected by $\pi_{\A}$ and let $\hat{S}_i$ be the first $i$ elements picked, i.e.,~$\hat{S}_i=\{\hat{u}_1,\dotsc,\hat{u}_i\}$, for $i=1,2,\dotsc,h$ where $h:=\abs{\N(\pi_{\A},\phi)}$. Suppose that there exists a partition $O_{1,1},O_{1,2},\dotsc,O_{1,h}$ of $O_1(\phi)$ such that for all $i=1,2,\dotsc,h$,
    \begin{equation}
        \sum\nolimits_{u\in O_{1,i}}\Delta(u\mid \psi_{u}(\phi))\leq k\cdot \Delta(\hat{u}_i\mid \psi_{\hat{u}_i}(\phi)),
    \end{equation}
    then Eqn.~\eqref{eqn:mustholdkeyeqn} must hold due to
    \begin{equation}
        \sum_{u\in O_1(\phi)} \Delta(u\mid \psi_{u}(\phi))= \sum_{i=1}^{h}\sum_{u\in O_{1,i}}\Delta(u\mid \psi_{u}(\phi))\leq  k\cdot\sum_{i=1}^{h} \Delta(\hat{u}_i\mid \psi_{\hat{u}_i}(\phi))=k\cdot \sum_{u\in \N(\pi_{\A},\phi)}\Delta(u\mid \psi_{u}(\phi)).
    \end{equation}
    Therefore, we just need to show the existence of such a desired partition of $O_1$, as proved below.

    \red{We use the following iterative algorithm to find the partition, which is inspired by \cite{calinescu2011maximizing}. Define $\N_{h}:=O_1(\phi)$. For $i=h,h-1,\dotsc,2$, let $B_i:=\{u\in \N_i\mid \hat{S}_{i-1}\cup \{u\}\in\I\}$. If $\abs{B_i}\leq k$, set $O_{1,i}=B_i$. Otherwise, pick an arbitrary $O_{1,i}\subseteq B_i$ with $\abs{O_{1,i}}=k$. Then, set $\N_{i-1}=\N_i\setminus O_{1,i}$. Finally, set $O_{1,1}=\N_1$. Clearly, $\abs{O_{1,i}}\leq k$ for $i=2,\dotsc,h$. We further show that $\abs{O_{1,1}}\leq k$. We prove it by contradiction and assume $\abs{O_{1,1}}> k$. If $\abs{B_2}\leq k$, then we have $\hat{S}_{1}\cup \{u\}\notin\I$ for every $u\in \N_1$ according to the above process. So $\hat{S}_1$ is a base of $\hat{S}_{1}\cup \N_1$, which implies that $\abs{\N_1}\leq k\cdot \abs{\hat{S}_1}$, contradicting the assumption that $\abs{\N_1}=\abs{O_{1,1}}> k$. Consequently, it must hold that $\abs{B_2}>k$ and hence $\abs{O_{1,2}}=k$ and $\abs{\N_2}>2k$. Using a similar argument, we can recursively get that $\abs{B_i}>k$ and hence $\abs{O_{1,i}}=k$ and $\abs{\N_i}>ik$ for any $i=3,\dotsc,h$, e.g., $\abs{\N_h}> hk$. However, as $\hat{S}_h$ is a base of $\hat{S}_h\cup O_1(\phi)$, we should have $\abs{\N_h}=\abs{O_1(\phi)}\leq hk$, which shows a contradiction. Therefore, we can conclude that $\abs{O_{1,i}}\leq k$ for all $i=1,2,\dotsc,h$.}

    According to the partition $O_{1,i}: i\in [h]$ constructed above, it is obvious that for every $u\in O_{1,i}$, $\hat{S}_{i-1}\cup \{u\}\in\I$. This implies that for every $u\in O_{1,i}$, $u$ cannot be considered before $\hat{u}_i$ is added by $\pi_{\A}$, i.e.,~$\psi_{\hat{u}_{i}}(\phi)\subseteq\psi_{u}(\phi)$. Meanwhile, due to the greedy rule of \ARG, it follows that $\Delta(\hat{u}_i\mid \psi_{\hat{u}_i}(\phi))\geq \Delta(u\mid \psi_{\hat{u}_{i}}(\phi))$ for each $u\in O_{1,i}$. Hence,
    \begin{equation}
        \sum_{u\in O_{1,i}} \Delta(u\mid \psi_{u}(\phi))\leq\sum_{u\in O_{1,i}} \Delta(u\mid \psi_{\hat{u}_{i}}(\phi))\leq\sum_{u\in O_{1,i}} \Delta(\hat{u}_i\mid \psi_{\hat{u}_i}(\phi)) \leq k\cdot \Delta(\hat{u}_i\mid \psi_{\hat{u}_i}(\phi))
    \end{equation}
    holds for any $i\in [h]$. Combining the above results completes the proof.
\end{proof}

\subsection{Proof of Lemma~\ref{lma:boundo2}}

\begin{proof}
    Again, since $f_{\avg}(\pi_{\A})= \E_{\pi_{\A}}\Big[\E_{\Phi} \Big[\sum_{u\in \N(\pi_{\A},\Phi)}\Delta(u\mid \psi_{u}(\Phi))\Big]\Big]$, we only need to prove that, for any $\phi\in Z^{\N}$,
    \begin{equation}
        \E_{\pi_{\A}}\Big[\sum\nolimits_{u\in O_2(\phi)} \Delta(u\mid \psi_{u}(\phi))\Big]\leq \frac{1-p}{p}\cdot \E_{\pi_{\A}}\Big[\sum\nolimits_{u\in \N(\pi_{\A},\phi)}\Delta(u\mid \psi_{u}(\phi))\Big].
    \end{equation}

    Given a realization $\phi\in Z^{\N}$, for each $u\in \N$, let $X_{u}$ be a random variable such that $X_{u}=1$ if $u\in O_2(\phi)$ and $X_{u}=0$ otherwise. So we have
    \begin{equation}
        \sum_{u\in O_2(\phi)} \Delta(u\mid \psi_{u}(\phi))=\sum_{u\in \N} \big(X_{u}\cdot \Delta(u\mid \psi_{u}(\phi))\big).
    \end{equation}
    Similarly, for each $u\in \N$, let $Y_{u}$ be a random variable such that $Y_{u}=1$ if $u\in \N(\pi_{\A},\phi)$ and $Y_{u}=0$ otherwise. Thus,
    \begin{equation}
        \sum_{u\in \N(\pi_{\A},\phi)}\Delta(u\mid \psi_{u}(\phi))=\sum_{u\in \N}\big(Y_{u}\cdot \Delta(u\mid \psi_{u}(\phi))\big).
    \end{equation}
    Therefore, it is sufficient to prove:
    \begin{equation}
        \forall u\in \N:~ \E_{\pi_{\A}}\big[X_{u}\cdot \Delta(u\mid \psi_{u}(\phi))\big]\leq \frac{1-p}{p}\cdot \E_{\pi_{\A}}\big[Y_{u}\cdot\Delta(u\mid \psi_{u}(\phi))\big]
    \end{equation}

    Observe that, for any given $u\in \N$, if $\Delta(u\mid \psi_{u}(\phi))\leq 0$ or $\dom(\psi_{u}(\phi))\cup\{u\}\notin \I$, then we have $u\notin \N(\pi_{\A},\phi)$ and $u\notin O_2(\phi)$ by definition, which indicates $X_{u}=Y_{u}=0$. Consider the event that $\Delta(u\mid \psi_{u}(\phi))> 0$ and $\dom(\psi_{u}(\phi))\cup\{u\}\in \I$, and denote such an event as $\mathcal{E}_{u}$. Since $\Pr[u\in \N(\pi_{\A},\phi)\mid \mathcal{E}_{u}]=p$, it is trivial to see that
    \begin{equation}
        \E_{\pi_{\A}}\big[Y_{u}\cdot\Delta(u\mid \psi_{u}(\phi))\big]=p\cdot \E_{\psi_{u}(\phi)}[\Delta(u\mid \psi_{u}(\phi))\mid \mathcal{E}_{u}]\cdot \Pr[ \mathcal{E}_{u}],
    \end{equation}
    where the expectation is taken over the randomness of $\psi_u(\phi)$ (i.e., $\psi_u(\phi)\sim \mathcal{E}_{u}$) due to the internal randomness of algorithm. On the other hand, if $u\in O(\phi)$, then we have $\Pr[u\in O_2(\phi)\mid \mathcal{E}_{u}]=1-p$ as $u$ is discarded with probability of $1-p$, while we also have $\Pr[u\in O_2(\phi)\mid \mathcal{E}_{u}]=0$ if $u\notin O(\phi)$. Thus, we know $\Pr[u\in O_2(\phi)\mid \mathcal{E}_{u}]\leq (1-p)$ and hence we can immediately get
    \begin{equation}
        \E_{\pi_{\A}}\big[X_{u}\cdot\Delta(u\mid \psi_{u})\big]\leq (1-p)\cdot \E_{\psi_{u}(\phi)}[\Delta(u\mid \psi_{u}(\phi))\mid \mathcal{E}_{u}]\cdot \Pr[ \mathcal{E}_{u}].
    \end{equation}
    The lemma then follows by combining all the above reasoning.
\end{proof}

\subsection{Proof of Theorem~\ref{thm:argapproxratio}}

\begin{proof}
    According to Lemmas \ref{lma:upperboundofmg}--\ref{lma:boundo2}, we have
    \begin{align*}
        f_{\avg}(\pi_{\A}@\pi_{\opt})-f_{\avg}(\pi_{\A})&\leq \E_{\pi_{\A},\Phi} \Big[\sum_{u\in \N(\pi_{\opt},\Phi)\setminus\N(\pi_{\A},\Phi)}\Delta(u\mid \psi_{u}(\Phi))\Big]  \\
        &\leq \E_{\pi_{\A},\Phi} \Big[\sum_{u\in O_1(\Phi)}\Delta(u\mid \psi_{u}(\Phi))+\sum_{u\in O_2(\Phi)}\Delta(u\mid \psi_{u}(\Phi))\Big] \\
        &\leq \big(k+\frac{1-p}{p}\big)\cdot f_{\avg}(\pi_{\A})
    \end{align*}
    where the second inequality is due to the definition of $O_3(\Phi)$, i.e., $\Delta(u\mid \psi_{u}(\Phi))\leq 0$ for every $u\in O_3(\Phi)$. Combining the above result with Lemma~\ref{lma:concat} gives
    \begin{equation}
	    f(\pi_{\opt})\leq \frac{kp+1}{p(1-p)}\cdot f_{\avg}(\pi_{\A}).
    \end{equation}
    Moreover, $\frac{kp+1}{p(1-p)}$ achieves its minimum value of $(1+\sqrt{k+1})^2$ at $p=(1+\sqrt{k+1})^{-1}$. Finally, the $\mathcal{O}(nr)$ time complexity is evident, as the algorithm incurs $\mathcal{O}(n)$ oracle queries for each selected element.
\end{proof}

%

\end{document}